\documentclass[a4paper, 11pt]{article}
\usepackage[english]{babel}
\usepackage{amsmath, amssymb, amsfonts, amsthm, bm}
\usepackage{mathrsfs}
\usepackage{enumerate}
\usepackage{float}
\usepackage{color,graphicx,epstopdf}
\usepackage{dsfont}
\usepackage{hyperref}
\usepackage{cite} 
\usepackage{algorithm} 
\usepackage{algorithmic} 
\usepackage[mathcal]{euscript}

\newcommand{\R}{\mathbb{R}}

\newcommand{\1}{\mathbf{1}}

\DeclareMathOperator{\rank}{rank}

\DeclareMathOperator{\var}{var}
\DeclareMathOperator{\cov}{cov}

\DeclareMathOperator*{\argmax}{argmax}
\DeclareMathOperator*{\argmin}{argmin}

\newcommand{\asf}{\mathcal{E}}

\newcommand{\vb}{\mathbf{v}}

\newcommand{\xb}{\mathbf{x}}
\newcommand{\yb}{\mathbf{y}}
\newcommand{\zb}{\mathbf{z}}

\newcommand{\etab}{\bm{\eta}}

\newcommand{\Hb}{\mathbf{H}}

\newcommand{\Ib}{\mathbf{I}}

\newcommand{\Pb}{\mathbf{P}}

\newcommand{\Ab}{\mathbf{A}}

\newcommand{\Cb}{\mathbf{C}}
\newcommand{\Db}{\mathbf{D}}

\newcommand{\db}{\mathbf{d}}

\newcommand{\qb}{\mathbf{q}}

\newcommand{\normdist}{\mathrm{N}}

\newcommand{\phib}{\bm{\phi}}

\newcommand{\Head}{\textnormal{Head}}
\newcommand{\Tail}{\textnormal{Tail}}

\newcommand{\tdropout}{t_{\mathrm d}}
\newcommand{\pdropout}{p_{\mathrm d}}

\newcommand{\mG}{\mathrm{G}}
\newcommand{\mV}{\mathrm{V}}
\newcommand{\mE}{\mathrm{E}}
\newcommand{\mT}{\mathrm{T}}
\newcommand{\mW}{\mathrm{W}}
\newcommand{\mU}{\mathrm{U}}

\newcommand{\mGrand}{\mathsf{G}}
\newcommand{\mVrand}{\mathsf{V}}
\newcommand{\mErand}{\mathsf{E}}

\newcommand{\mTG}{\mT^{^o}_{\mG}}
\newcommand{\mET}{\mE_{\mT}^{^o}}

\newcommand{\mGM}{\mathpzc{G}}
\newcommand{\mEM}{\mathpzc{E}}

\newcommand{\gammab}{{\bm \gamma}}
\newcommand{\lambdab}{{\bm \lambda}}

\newcommand{\mS}{\mathrm{S}}

\newcommand{\tbeg}{n-1 }
\newcommand{\tend}{t^\ast}

\newcommand{\xib}{\bm{\xi}}
\newcommand{\betab}{\bm{\beta}}
\newcommand{\mub}{\bm{\mu}}
\newcommand{\Lambdab}{\bm{\Lambda}}
\newcommand{\Sigmab}{\bm{\Sigma}}
\newcommand{\Phib}{\bm{\Phi}}

\newcommand{\betamlefs}{{\betab}_{\textnormal{mle}}}

\newcommand{\betamapfs}{{\betab}_{\textnormal{map}}}

\newcommand{\emlefs}{\mathbf{E}_{\textnormal{mle}}}

\newcommand{\emapfs}{\mathbf{E}_{\textnormal{map}}}

\newcommand{\bmlefs}{\mathbf{b}_{\textnormal{mle}}}

\newcommand{\bmapfs}{\mathbf{b}_{\textnormal{map}}}

\newcommand{\mR}{\mathrm{R}}

\newcommand{\sigm}{\sigma_{\rm m}}

\DeclareFontFamily{OT1}{pzc}{}
\DeclareFontShape{OT1}{pzc}{m}{it}{<-> s * [1.200] pzcmi7t}{}
\DeclareMathAlphabet{\mathpzc}{OT1}{pzc}{m}{it}

\newcommand*\mcap{\mathbin{\mathpalette\mcapinn\relax}}
\newcommand*\mcapinn[2]{\vcenter{\hbox{$\mathsurround=0pt
			\ifx\displaystyle#1\textstyle\else#1\fi\bigcap$}}}

\newcommand*\mcupinn[2]{\vcenter{\hbox{$\mathsurround=0pt
			\ifx\displaystyle#1\textstyle\else#1\fi\bigcup$}}}

\newtheorem{proposition}{Proposition}
\newtheorem{theorem}{Theorem}
\newtheorem{lemma}{Lemma}

\newtheorem{definition}{Definition}

\newtheorem{remark}{Remark}

\title{\bf Dynamical Privacy in  Distributed Computing\\
	Part II:  PPSC Gossip Algorithms}
\date{}
\author{Yang Liu, Junfeng Wu, Ian Manchester, Guodong Shi\thanks{Y. Liu is with the Research School of Engineering, The Australian National University, Canberra 0200, Australia. (email: yang.liu@anu.edu.au).}%
\thanks{J. Wu is with   with the College of Control Science and Engineering, Zhejiang
University, Hangzhou 310027, China (e-mail: jfwu@zju.edu.cn)}%
\thanks{I. R. Manchester is with Australian Center for Field Robotics, The University of Sydney, NSW 2006, Australia. (email: ian.manchester@sydney.edu.au)}
\thanks{G. Shi is with Australian Center for Field Robotics, The University of Sydney, NSW 2006, Australia, and the Research School of Engineering, The Australian National University, Canberra 0200, Australia (email: guodong.shi@sydney.edu.au)}
 }

\begin{document}
\maketitle

\begin{abstract}
	In the first part of the paper, we have studied the computational privacy risks in distributed computing protocols against local or global dynamics eavesdroppers, and proposed   a Privacy-Preserving-Summation-Consistent (PPSC) mechanism   as a    generic privacy encryption subroutine for  consensus-based distributed computations.  In this part of this paper, we show that the conventional deterministic and random  gossip algorithms can be used to realize the PPSC mechanism over a given network. At each time step, a node is selected to interact with one of its neighbors via deterministic or random gossiping. Such node generates a random number as its new state,  and  sends the subtraction between its current state and that random number to the neighbor; then the neighbor updates its state by adding the received value to its current state. We establish  concrete privacy-preservation conditions by proving the impossibility  for the reconstruction of the network input from the output of the gossip-based  PPSC mechanism against eavesdroppers with full network knowledge, and by showing  that the PPSC mechanism can achieve differential privacy at arbitrary privacy levels.  The convergence is characterized  explicitly and analytically for both deterministic and randomized gossiping, which is essentially achieved in a finite number of steps. Additionally, we illustrate that the proposed algorithms can be easily made adaptive in real-world applications by making realtime trade-offs between resilience against  node dropout or communication failure and privacy preservation  capabilities.
\end{abstract}

\section{Introduction}

The development of distributed control and optimization  has become one of the central  streams in the study of complex network operations,  due to the rapidly growing volume and dimension of data and information flow from a variety of applications  such as social networking, smart grid, and intelligent transportation \cite{magnusbook,nedic09,nedic10,boyd2011distributed,xiaoming-formation,khan2009distributed}. Consensus algorithms serve as a basic tool for the information dissemination of distributed algorithms \cite{jad03,xiao04}, where  a network aim to achieve a common value related to the network initials by local interactions, usually just the average. This distributed nature of consensus algorithms allows for   optimization of sum of  local convex objective functions by gradient descent \cite{nedic10}, and  solving  linear algebraic equations over networks with equations allocated at the individual nodes   \cite{Mou-TAC-2015,Shi-TAC-2017}.
   The origin of this line of research can be traced back to the 1980s from the work of distributed optimization and decision making \cite{tsi} and parallel computations \cite{lynch}.

Under a distributed computing structure, the 
information about the network-wide computation problem is encoded in the individual node initial values or update rules, which are not shared directly among the nodes. Rather, nodes share certain dynamical states based on neighboring communications and local problem datasets, in order to find solutions of the global problem. The node dynamical states may contain sensitive private information directly or indirectly, but  in distributed computations to achieve the network-wide computation goals nodes   have to accept the fact that their privacy, or at least some of it, will be inevitably lost. It was pointed out that indeed nodes  lose their  privacy in terms of initial values if an attacker, a malicious user, or an eavesdropper knows part of the node trajectories and the network structure \cite{sundaram2007,yuan2013}. In fact, as long as observability holds, the whole network initial values become re-constructable with only a segment of node state updates at a few selected nodes.  Several insightful privacy-preserving consensus algorithms have been presented in the literature \cite{huang2012differentially,manitara2013privacy,mo2017privacy,nicolas2017,claudio2018} in the past few years, where the central idea is to inject  random noise or offsets in the node communication and iterations, or employ local maps as node state masks. The classical notion on differetial privacy in computer science has also been introdiuced to  distributed optimization problems   \cite{huang2015differentially,pappas2017}, following the earlier work on differentially private filtering \cite{pappas2014}.

In the first part of the paper, we have demonstrated the computational risks in distributed computation  algorithms and particularly in distributed linear equation solvers, and then defined the  
PPSC mechanism and illustrated its usefulness to enhance privacy in distributed computing. In this second part the paper, we focus on  the realization of PPSC mechanism  based on gossip algorithms, and the resulting privacy-preserving promises. A gossip process takes place over a network, where  a  pair of nodes is selected randomly or deterministically at any given time, and then this pair of nodes {\em gossip} in the sense that they exchane  information between each other as  a fundamental resource allocation protocol for computer networks \cite{Demers1987,Kempe2004}. The interaction process of a gossiping pair can be modeled as averaging, leading to  the definition of gossip algorithms \cite{kempe2003} serving as a description to an underlying gossip process or a distributed computing algorithm in its own right. Protocols based on gossiping have led to new distributed networking solutions in large-scale applications \cite{shah2008,boyd2006randomized,nicolas2016} and new models of social network evolutions \cite{doerr-2012}.

We show that standard gossip algorithms can be simply modified to realize PPSC mechanism in finite time steps. The main results for Part II of this paper is summarized as follows. 
\begin{itemize}
	\item Gossip-based PPSC algorithms as realizations of PPSC mechanism are proposed. At each time step,  one node is selected to  initialize an interaction with its neighbor via deterministic or random gossiping. This node generates a random number as its new state,  and  sends the subtraction between its current state and that random number to the neighbor. Then the neighbor updates its state by adding the received value to its current state. 
	
	\item The resulting PPSC mechanism is shown to have proven privacy guarantee in the sense that upon observing the output of the mechanism, the input is not identifiable. Moreover, we can always design the setup of the gossip algorithm to achieve differential privacy at arbitrary privacy budge levels for the PPSC mechanism. 
	
	\item The convergence rate and convergence limits of both the deterministic and randomized PPSC algorithms are explicitly established. 
\end{itemize}
We note that the proposed  PPSC mechanism has the following dinstinct feature compared to exisiting work on privacy preservation \cite{huang2012differentially,manitara2013privacy,mo2017privacy,nicolas2017,claudio2018}: even if we completely eliminate the added random numbers, the outcome of the mechanism continues to provide a degree of network privacy protection in the sense of  non-identifiability. In fact, we show that the random numbers primarily ensure local privacy protection and further enhance network privacy preservation measured by differential privacy.

A preliminary version of the results are to be reported at IEEE Conference on Decision and Control in Dec. 2018 \cite{cdc1,cdc2}. In the current manuscript, we have established a series of   new results on privacy loss characterizations and privacy preservation quantifications, in addition to a few new illustration  examples and technical proofs. 

 The remainder of the paper is organized as follows. In Section \ref{sec:ppsc}, we recall the definition of the PPSC mechanism and some preliminary knowledge on graph theory.   The deterministic privacy-preserving algorithm is proposed and investigated  in Section \ref{sec:deter}.  The randomized privacy-preserving algorithms are presented in Section \ref{sec:rand}, in which the trade-off between resilience and privacy preservation is also analyzed. Some concluding remarks are finally drawn in Section \ref{sec:conclusion}.

\section{Preliminaries}\label{sec:ppsc}
\subsection{PPSC Mechanism}
Consider a network of $n$ nodes  indexed in the set $\mathrm{V}=\{1,\dots,n\}$ and interconnected according to a  graph $\mathrm{G}=(\mathrm{V},\mathrm{E})$, where $\mathrm{E}$ is a set of bidirectional edges between the nodes in $\mathrm{V}$. Each node $i\in\mV$ holds an initial state $\beta_i\in\R^r$.  An algorithm running over the node set $\mathrm{V}$ is called a distributed Privacy-Preserving-Summation-Consistent (PPSC) algorithm, which produces output $\betab^\sharp=(\beta_1^{\sharp}\ \dots\ \beta_n^{\sharp})^\top$ from the network input $\betab=(\beta_1\ \dots\ \beta_n)^\top$, if the following four conditions hold:
\begin{itemize}
	\item [(i)] {\em (Graph Compliance)} Each node $i$ communicates only  with its neighbors in the set $\mathrm{N}_i:=\big\{j:\{i,j\}\in\mathrm{E}\big\}$;
	\item[(ii)] {\em (Local Privacy Preservation)} Each node $i$ never reveals its initial value $\beta_i$ to any other agents or a third party;
	\item[(iii)] {\em (Network Privacy Preservation)} $\betab$ is non-identifiable given $\betab^\sharp$;
	\item[(iv)] {\em (Summation Consistency)} $\sum\limits_{i=1}^n \beta^{\sharp}_i=\sum\limits_{i=1}^n \beta_i$.
\end{itemize}

\subsection{Graph Theory Preliminaries}
Consider an undirected graph $\mG=(\mV,\mE)$.
Nodes $i$ and $j$ are adjacent if $\{i,j\}\in\mE$.
A sequence of distinct nodes $i_0,i_1,\dots,i_l$ is said to be a path of length $l\ge 1$ between node $i_0$ and $i_l$ if $i_j$ and $i_{j+1}$ are adjacent for all $j=0,1,\dots,l-1$.
A graph $\mG$ is  connected if there exists at least one path between $i$ and $j$ for any $i\neq j\in\mV$.
A spanning subgraph of $\mathrm{G}$ is defined as a graph with its node set being $\mV$ and its edge set being a subset of $\mE$. Then we say $\mathrm{T}_{\mG}=(\mathrm{V},\mathrm{E}_{\mT})$ is a spanning tree of connected $\mathrm{G}=(\mathrm{V},\mathrm{E})$ if $\mathrm{T}_{\mG}$ is a spanning subgraph of $\mathrm{G}$ and is a tree. An orientation over the edge set $\mE$ is  a mapping  $o:\mE\to\{-1,1\}$ with $o(i,j)=-o(j,i)$ for all $\{i,j\}\in\mathrm{E}$. A directed edge, denoted by an ordered pair of nodes $(i,j)$, is generated under this orientation $o$ if $o(i,j)=1$. Particularly, $i$ is the tail of the directed edge $(i,j)$, denoted by $\Tail\big((i,j)\big)$; and $j$ is the head denoted by $\Head\big((i,j)\big)$. The graph $\mG$ with an orientation $o$ results in a directed graph $\mG^{^o}=(\mV,\mE^{^o})$ with $\mE^{^o}=\{(i,j):o(i,j)=1,\ \forall\{i,j\}\in\mE\}$, and in turn a sequence of distinct nodes $i_0,i_1,\dots,i_l$ is said to be a directed path of length $l\ge 1$ if $(i_j,i_{j+1})\in\mE^{^o}$ for all $j=0,1,\dots,l-1$. We refer to \cite{godsil2013} for more details on graph theory.

\section{Deterministic PPSC Gossiping}\label{sec:deter}
In this section, we present a deterministic  gossip-based PPSC algorithm. To carry out a deterministic gossiping  process over a network, one has to arrange the sequence of node interactions, which is often a hard constraint in practical environments.  However, the study of deterministic gossip protocols can eliminate  the randomness of the gossiping process, and therefore allow for analysis  focusing  on the inherent randomness of the node updates themselves. Therefore, deterministic gossip algorithms often serve well as benchmarks for the performance of the general category of gossip algorithms \cite{shi-tit-2015}.

\subsection{The Algorithm}
We assume $r=1$ to simplify the discussion, i.e., $\beta_i \in \mathbb{R}$ without loss of generality,  because for $\beta_i\in\mathbb{R}^r$ with $r>1$,  the presented algorithm   can run component-wise along each entry.  Let  $\mathrm{T}_{\mG}=(\mV,\mE_{\mT})$ be any spanning tree of $\mG$. Then we assign an arbitrary  orientation $o$ to $\mathrm{T}_{\mG}$ so that an oriented tree $\mTG=(\mV,\mET)$  is obtained.   We sort the directed edges in $\mET$ by $\asf_1,\asf_2,\dots,\asf_{n-1}$ with an arbitrary order. Let $\gammab_t, \ t=1,2,\dots$ be independently and identically distributed (i.i.d.) random variables.
\begin{algorithm}[htb]
	{$\mathbf{Deterministic\textnormal{ }PPSC\textnormal{ }Gossip\textnormal{ }Algorithm}$}\\
	Set $x_i(0)\gets\beta_i$ for each $i\in\mV$. For $t=1,\dots,n-1 $, repeat the following steps.
	\begin{algorithmic}[1]
		\STATE Node $\Tail(\asf_t)$ randomly and independently generates $\gammab_t\in\R$ according to some distribution with  mean  $\varrho_\gammab$ and  variance $\sigma_{\gammab}^2>0$.
		\STATE Node $\Tail(\asf_t)$ computes $\omega_t=x_{\Tail(\asf_t)}(t-1)-\gammab_t$ and sends $\omega_t$ to $\Head(\asf_t)$.
		\STATE Node $\Tail(\asf_t)$ updates its state by $x_{\Tail(\asf_t)}(t)=\gammab_t$, node $\Head(\asf_t)$ updates its state by $x_{\Head(\asf_t)}(t)=x_{\Head(\asf_t)}(t-1)+\omega_t$, and each node $i\in\mV\setminus\{\Tail(\asf_t),\Head(\asf_t)\}$ sets $x_i(t)=x_i(t-1)$.
	\end{algorithmic}
\end{algorithm}

After $n-1$ rounds, this deterministic PPSC (D-PPSC) gossip algorithm produces 
$$
\beta_i^\sharp=x_i(n-1 ), \ i=1,\dots,n
$$
from the initial condition $x_i(0)=\beta_i,\ i=1,\dots,n$. At each time $t$, only node $\Tail(\asf_{t})$ sends
$
\omega_t = x_{\Tail(\asf_{t})}(t-1) - \gammab_t
$
to node $\Head(\asf_{t})$. Therefore, we can verify   that the D-PPSC-Gossip algorithm satisfies the Graph Compliance and Local Privacy Preservation conditions. The Summation Consistency condition can also be easily verified since the sum of the node states for the interacting pair in a step never changes.

The D-PPSC-Gossip algorithm can have a simple algebraic representation. We define $\Phib_k\in\R^{n\times n}$ for $k=1,2,\dots,n-1 $ by $[\Phib_k]_{ij}=1$ {if} $i=\Head(\asf_k)\textnormal{ and }j=\Tail(\asf_k)$; $[\Phib_k]_{ij}=-1$ {if} $i=\Tail(\asf_k)\textnormal{ and }j=\Tail(\asf_k)$; and  $[\Phib_k]_{ij}=0$ otherwise.  
Then along the D-PPSC-Gossip algorithm, there holds
\begin{equation}\label{eq:linear_noninitial_initial}
\begin{aligned}
\xb(1)&=\Ab_1\betab+\gammab_1\vb_1,\\
&\vdots\\
\xb(n-1 )&=\Ab_{n-1 }\xb(n-1 -1)+\gammab_{n-1 }\vb_{n-1 },
\end{aligned}
\end{equation}
where $\Ab_k=\Phib_k+\Ib$, and $\vb_k$ is  the all-zeros vector except for its $\Tail(\asf_k)$-th component being $1$ and its $\Head(\asf_k)$-th component being $-1$ for $k=1,\dots,n-1 $. Introduce $$\Cb=\Ab_{n-1 }\cdots\Ab_1 
$$ and $$
\Db=[\Ab_{n-1 }\cdots\Ab_2\vb_1,\ \Ab_{n-1 }\cdots\Ab_3\vb_2,\ \dots,\ \vb_{n-1 }]  
$$ Denote $\gammab=[\gammab_1\ \dots\ \gammab_{n-1 }]^\top$. Then (\ref{eq:linear_noninitial_initial}) yields
\begin{equation}\label{eq:linear_final_initial}
\betab^\sharp=\mathscr{P}_{\rm D-PPSC}(\betab)=\Cb\betab+\Db\gammab
\end{equation}
where $\Cb\in \mathbb{R}^{n\times n}$, $\Db\in \mathbb{R}^{n\times (n-1)}.$

The matrices $\mathbf{A}_k$ have very special structures, and in turns out that  we can prove from such structures  that $\mathbf{C}$ is always singular. As a result, even if $\gammab=0$, i.e., no randomness is injected to $\betab^\sharp$, knowing $\mathbf{C}$ and $\betab^\sharp$ is not enough to determine $\betab$. As a result, the PPSC mechanism by itself provides a degree of network privacy preservation. On the other hand,  obviously the existence of $\gammab$ makes the local privacy preservation  possible. We will further  show that $\gammab$   enhances the network privacy protection measured under differential privacy.

\begin{remark}{\rm 
Note that the D-PPSC-Gossip algorithm can run the same way over an arbitrary number of sorted  edges with orientation, not necessarily relying on a spanning tree and the resulting $n-1$ steps. The spanning tree and the $(n-1)$-step setup are just to make sure that all nodes have updated their input values, which are not essential to the privacy protection as shown in the analysis below.  }
\end{remark}


\subsection{Global Dynamics Privacy}
Recall that  $\betab=[\beta_1\ \dots\ \beta_n]^\top$ and $\betab^\sharp=[\beta_1^\sharp\ \dots\ \beta_n^\sharp]^\top$. Let the randomized  mechanism mapping $\betab$ to $\betab^\sharp$ along the  D-PPSC-Gossip algorithm
be denoted by 
\begin{align}
\betab^\sharp= \mathscr{P}_{\rm D-PPSC} (\betab). 
\end{align}
\begin{theorem}\label{thm0}
	The mechanism  $\mathscr{P}_{\rm D-PPSC}$ by the D-PPSC-Gossip algorithm preserves network privacy in the sense that $\betab$ is always non-identifiable,  by eavesdroppers knowing the network topology $\mG$, the spanning tree $\mT_\mG$, the oriented spanning tree $\mTG$,   the edge sequence $(\asf_1,\dots,\asf_{n-1 })$, and the statistics or even distribution  of the $\gammab_t$. 
\end{theorem}
The proof of Theorem   \ref{thm0} can be found in the appendix. 
There are several points worth emphasizing. 
\begin{itemize}
	\item  The direct implication of Theorem \ref{thm0} is that the mechanism  $\mathscr{P}_{\rm D-PPSC}$ by the D-PPSC-Gossip algorithm  is privacy preserving by that for any fixed $\betab$, it is impossible to recover  $\betab$ even with an infinite number of realization of  $\mathscr{P}_{\rm D-PPSC} (\betab)=\betab^\sharp$ for a fixed spanning tree structure.
	
	\item In the meantime,   any realization $\mathscr{P}_{\rm D-PPSC} (\betab)=\betab^\sharp$   contains information about $\betab$. For example, with one single realization one immediately learns $\mathbf{1}^\top \betab =\sum_{i=1}^n \beta
	_i$. This is inevitable because otherwise  $\mathscr{P}_{\rm D-PPSC} (\betab)$ will not be useful at all for computation based on $\betab$.
	\item Such recovering impossibility holds true for eavesdroppers having the full knowledge of the D-PPSC-Gossip algorithm setup: the network topology $\mG$, the spanning tree $\mT_\mG$, the oriented spanning tree $\mTG$,   the edge sequence $(\asf_1,\dots,\asf_{n-1 })$, and the statistics or even distribution  of the $\gammab_t$. Missing any piece of such knowledge will make it much harder to learn more information about $\betab$ from $\betab^\sharp$.  
\end{itemize}
We also remark that nodes may certainly  lose privacy in terms of $\betab$ during the recursion of the D-PPSC-Gossip algorithm against eavesdroppers that have access to the node states $x_i(t)$ and communication packets $\omega_t$. However, first of all, one can ensure the security of  the PPSC algorithm so that the output  $\betab^\sharp$ can be made public as the privacy of $\mathscr{P}_{\rm D-PPSC}$ is guaranteed by Theorem \ref{thm0}. Secondly, the privacy of the PPSC mechanism can be significantly improved by a randomized gossiping process, which will be discussed later. A randomized PPSC algorithm will be able to hide   the oriented spanning tree $\mTG$ and  the edge sequence $(\asf_1,\dots,\asf_{n-1 })$ through simple randomization.

\subsection{Differential Privacy}

We introduce the following definition on the differential privacy of the D-PPSC-Gossip algorithm. 
\begin{definition}
	(i) Two inputs for the PPSC mechanism $\betab,\betab^\prime\in\R^n$ are  $\delta$--adjacent if $\1^\top\betab=\1^\top\betab^\prime$ and there exists a unique $i\in\mV$ such that (i) $\left|\beta_i-\beta^\prime_i\right|\le\delta$; (ii) $\beta_j=\beta^\prime_j$ for all $j\neq i$.
	
	(ii)
	The D-PPSC-Gossip algorithm preserves $\epsilon $--differential privacy under $\delta$-adjacency  if for all $\mR\subset{\rm range}(\mathscr{P}_{\rm D-PPSC})$, there holds
	$$\Pr\big(\mathscr{P}_{\rm D-PPSC}(\betab)\in\mR\big)\le e^\epsilon\cdot\Pr\big(\mathscr{P}_{\rm D-PPSC}(\betab^\prime)\in\mR\big)$$
	for two arbitrary $\delta$--adjacent initial conditions $\betab,\betab^\prime\in\R^n$.
\end{definition}

Let $\Sigmab_{\betab^\sharp}$ be the covariance matrix\footnote{We view $\betab$ as a deterministic input here and in fact for the majority part of the paper, i.e., the randomness for this covariance comes entirely from the random numbers $\gamma_t$. } of $\betab^\sharp=\mathscr{P}_{\rm D-PPSC}(\betab)$. We define $\mGM_{\betab^\sharp}$ as the graphical model of $\betab^\sharp$ with $\mGM_{\betab^\sharp}=(\mV,\mEM_{\betab^\sharp})$, where $\mEM_{\betab^\sharp}=\big\{\{i,j\}:i\neq j,\ [\Sigmab_{\betab^\sharp}]_{ij}\neq 0\big\}$. Note that this graphical model $\mGM_{\betab^\sharp}$ does not depend on specific values of $\betab$. We will establish shortly in the next subsection the structural property of this graphical model $\mGM_{\betab^\sharp}$. 

Let $\Delta(\cdot)$ denote the maximum degree of a graph, and $\sigm(\cdot)$ denote the smallest eigenvalue in absolute value of a real symmetric matrix. We establish the following result which further quantifies the degree of network privacy preservation by the D-PPSC-Gossip algorithm. 
\begin{theorem}\label{thmdp}
	Suppose $\gammab_t$ with $t=1,\dots,n-1 $ are i.i.d. Laplace random variables with variance $2v^2,\ v>0$. Then the   mechanism $\mathscr{P}_{\rm D-PPSC}$ is $\epsilon$-differentially private under $\delta$-adjacency if 
	$$
	\frac{\delta\sqrt{n-1}\Delta(\mGM_{\betab^\sharp})}{v\left|\sigm(\Db^\top\Db)\right|}\le\epsilon.  
	$$
\end{theorem}

Theorem \ref{thmdp} shows that the PPSC mechanism can provide $\epsilon$-differential privacy  for all $\epsilon$ as long as the variance of the $\gammab_t$ is sufficiently large. Since the sum is consistent regardless of the choice of $\gammab_t$, the computation accuracy can also be guaranteed. Together with Theorem \ref{thm0}, it can be seen that $\mathscr{P}_{\rm D-PPSC}(\betab)$ provides each node with strong deniability about holding $\beta_i$ upon observing $\betab^\sharp$. The proof of Theorem   \ref{thmdp} can be found in the appendix.

\begin{remark}{\rm 
The generated  PPSC mechanism  $\mathscr{P}_{\rm D-PPSC} $  certainly depends on the choice of the tree  $\mathrm{T}_{\mG}$. The dependence is encoded in the representations of $\mathbf{C}$, $\mathbf{D}$, and the graphical model $\mGM_{\betab^\sharp}$, which in turn affects the privacy preserving degree of  $\mathscr{P}_{\rm D-PPSC} $. Therefore, we may further optimize   the choice of   $\mathrm{T}_{\mG}$ for the $\mathscr{P}_{\rm D-PPSC} $ at the synthesis stage of the algorithm. }
\end{remark}

\subsection{Algorithm  Output Statistics}
 It is of interest to understand the statistical relations between the different  $\beta^\sharp_i$'s under the mechanism $\mathscr{P}_{\rm D-PPSC}$.  We denote the sequantial order of two directed edges  in $\mathrm{T}_{\mG}$ by $\asf_s\prec\asf_t$ if $s<t$ and $\asf_s\preceq\asf_t$ if $s\le t$, and write $\min(\asf_s,\asf_t)=\asf_{\alpha}$ with $\alpha=\min(s,t)$, and $\max(\asf_s,\asf_t)=\asf_{\alpha}$ with $\alpha=\max(s,t)$.
The following theorem proposes a necessary and sufficient condition for the dependence of two nodes' final states in the case when the two nodes are not directly connected by a path.

\medskip

\begin{theorem}\label{thm:dependence}
	Suppose there exists no directed path in $\mTG$ connecting node $i$ and node $j$. Let $i_0,i_1,\dots,i_l$ denote the unique undirected path in $\mathrm{T}_{\mG}$ that connects node $i$ and node $j$ with $i=i_0$ and $j=i_l$. Then along the D-PPSC-Gossip algorithm, $\beta^\sharp_{i_0}$ and $\beta^\sharp_{i_l}$ are dependent if and only if there exists $0<p<l$ such that the following conditions hold.
	\begin{enumerate}[(i)]
		\item $i_p,i_{p-1},\dots,i_{0}$ and $i_p,i_{p+1},\dots,i_l$ are both directed paths.
		\item There hold $(i_p,i_{p-1})\prec\dots\prec(i_1,i_0)$ and $(i_p,i_{p+1})\prec\dots\prec(i_{l-1},i_l)$.
		\item \begin{itemize}
			\item[a)]If $(i_0,i^\ast)\in\mET$, then $(i_0,i^\ast)\prec(i_0,i_{1})$;
			\item[b)]If $(i_b,i^\ast)\in\mET$ for $0<b<p$, then $(i_b,i^\ast)\prec(i_{b+1},i_{b})\textnormal{ or }(i_b,i_{b-1})\preceq(i_b,i^\ast)$;
			\item[c)] If $(i_p,i^\ast)\in\mET$, then $$(i_p,i^\ast)\prec\min\big((i_p,i_{p-1}),(i_p,i_{p+1})\big)$$ or $$\max\big((i_p,i_{p-1}),(i_p,i_{p+1})\big)\preceq(i_p,i^\ast);$$
			\item[d)] If $(i_b,i^\ast)\in\mET$ with $p<b<l$, then $(i_b,i^\ast)\prec(i_{b-1},i_{b})\textnormal{ or }(i_b,i_{b+1})\preceq(i_b,i^\ast)$;
			\item[e)] If $(i_l,i^\ast)\in\mET$, then
			$(i_l,i^\ast)\prec(i_{l-1},i_l)$.
		\end{itemize}
	\end{enumerate}

	
	
	
	
\end{theorem}

\medskip

{Intuitively Condition (i) and Condition (ii) in Theorem \ref{thm:dependence} constrain the directions and selection sequences of the oriented edges on the path $i_0,\dots,i_l$. It is worth noting that Condition (iii) further characterizes the sequential order of selecting the edges on the path $i_0,\dots,i_l$, and the edges with one endpoints in the path $i_0,\dots,i_l$.} The following theorem provides a necessary and sufficient condition for the dependence of two arbitrary nodes' final states when there exists a directed path connecting them.

\medskip

\begin{theorem}\label{thm:dependence_directed_path}
	Suppose $i_0,i_1,\dots,i_l$ is a directed path of length $l$ in $\mTG$.  Then along the D-PPSC-Gossip algorithm there holds that  $\beta^\sharp_{i_0}$ and $\beta^\sharp_{i_l}$ are dependent if and only if the following conditions hold:
	\begin{enumerate}[(i)]
		\item $(i_0,i_{1})\prec(i_1,i_2)\prec\dots\prec(i_{l-1},i_l)$.
		\item
		\begin{itemize}
			\item[a)] $(i_0,i^\ast)\prec(i_0,i_{1})$ when $(i_0,i^\ast)\in\mET$;
			\item[b)] $(i_b,i^\ast)\prec(i_{b-1},i_{b})\textnormal{ or }(i_b,i_{b+1})\prec(i_b,i^\ast)$ when $(i_b,i^\ast)\in\mET$;
			\item[c)] $(i_l,i^\ast)\prec(i_{l-1},i_l)$ when $(i_l,i^\ast)\in\mET$.
		\end{itemize}
	\end{enumerate}
\end{theorem}

%

For the graphical model $\mGM_{\betab^\sharp}$,   the following theorem holds.

\begin{theorem}\label{thm:Sigma}
	The graphical model $\mGM_{\betab^\sharp}$ is a tree. Moreover, $\Sigmab_{\betab^\sharp}/\sigma_{\gammab}^2$ is the Laplacian of $\mGM_{\betab^\sharp}$.
\end{theorem}

The proofs of Theorem \ref{thm:dependence} -- Theorem \ref{thm:Sigma} are deferred to the appendix. 
\subsection{Examples}
We now present a few illustrative examples.

\noindent{\bf Example 1.} Consider a $5$-node undirected graph $\mG=(\mV,\mE)$ as shown in Figure \ref{fig:conv_ori}. We select an oriented spanning tree $\mTG=(\mV,\mET)$ as in Figure \ref{fig:conv}. Let the edges in $\mET$ be sorted as
\begin{equation}\notag
\asf_1=(5,2)\prec\asf_2=(2,3)\prec\asf_3=(2,1)\prec\asf_4=(3,4).
\end{equation}

\begin{figure}[H]
	\centering
	\includegraphics[width=2in]{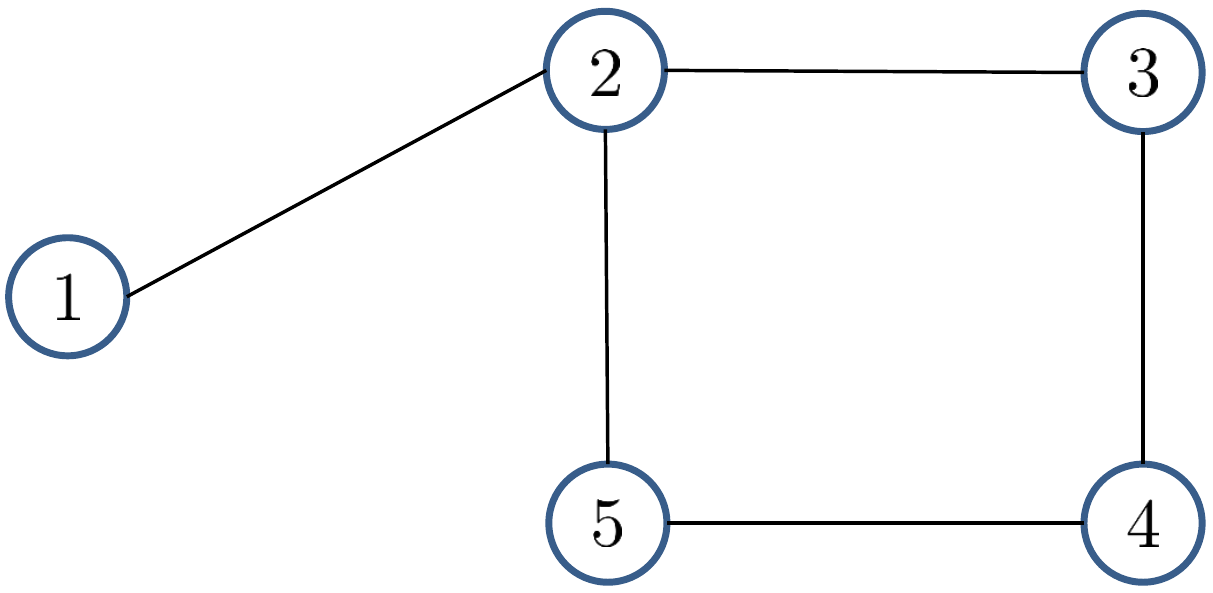}
	\caption{A $5$-node undirected graph $\mG$.}
	\label{fig:conv_ori}
\end{figure}

\begin{figure}[H]
	\centering
	\includegraphics[width=2in]{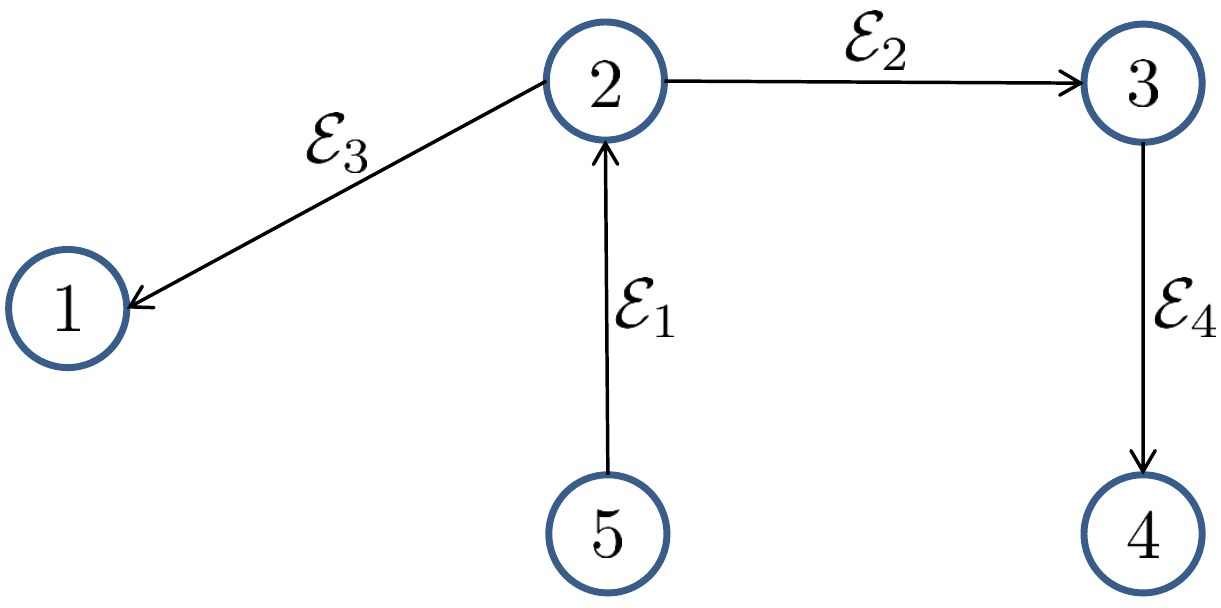}
	\caption{An oriented spanning tree $\mTG$ of $\mG$.}
	\label{fig:conv}
\end{figure}

The algorithm starts with $x_i(0)=\beta_i,\ i=1,\dots,5$ and produces
\begin{equation}\notag
\betab^\sharp=
\begin{bmatrix}
\beta_1+\gammab_2-\gammab_3 \\
\gammab_3 \\
\gammab_4 \\
\beta_2+\beta_3+\beta_4+\beta_5-\gammab_1-\gammab_2-\gammab_4 \\
\gammab_1
\end{bmatrix}.
\end{equation}
Evidently $\sum\limits_{i=1}^5 \beta^\sharp_i=\sum\limits_{i=1}^5 \beta_i$, i.e., network node states sum is preserved.
In addition, one can see that the following conditions hold: (i) $2,1$ and $2,3,4$ are both directed paths in $\mTG$; (ii) $(2,3)\prec(3,4)$; (iii) There exists no node $i^\ast$ such that $(i,i^\ast)\in\mET$ for any $i=1,2,3,4$.
This shows that the node state pair $(\beta^\sharp_1,\beta^\sharp_4)$ satisfies the dependence conditions of Theorem \ref{thm:dependence}. Clearly $\beta^\sharp_1$ and $\beta^\sharp_4$ are dependent, which validates Theorem \ref{thm:dependence}. Validation of Theorem \ref{thm:dependence_directed_path} can be similarly shown from the node state pair $\beta^\sharp_4$ and $\beta^\sharp_5$.

The graphical model $\mGM_{\betab^\sharp}$ of $\betab^\sharp$ is illustrated in Figure \ref{fig:conv_graphical_model}. By direct calculation, we have
\begin{equation}\notag
\Sigmab_{\betab^\sharp}=
\begin{bmatrix}
2\sigma_\gammab^2 & -\sigma_\gammab^2 & 0 & -\sigma_\gammab^2 & 0 \\
-\sigma_\gammab^2 & \sigma_\gammab^2 & 0 & 0 & 0 \\
0 & 0 & \sigma^2 & -\sigma_\gammab^2 & 0 \\
-\sigma_\gammab^2 & 0 & -\sigma_\gammab^2 & 3\sigma_\gammab^2 & -\sigma_\gammab^2 \\
0 & 0 & 0 & -\sigma_\gammab^2 & \sigma_\gammab^2
\end{bmatrix}.
\end{equation}
It is clear that $\Sigmab_{\betab^\sharp}/\sigma_\gammab^2$ is the Laplacian of $\mGM_{\betab^\sharp}$ and this validates Theorem \ref{thm:Sigma}.

\begin{figure}[H]
	\centering
	\includegraphics[width=2.5in]{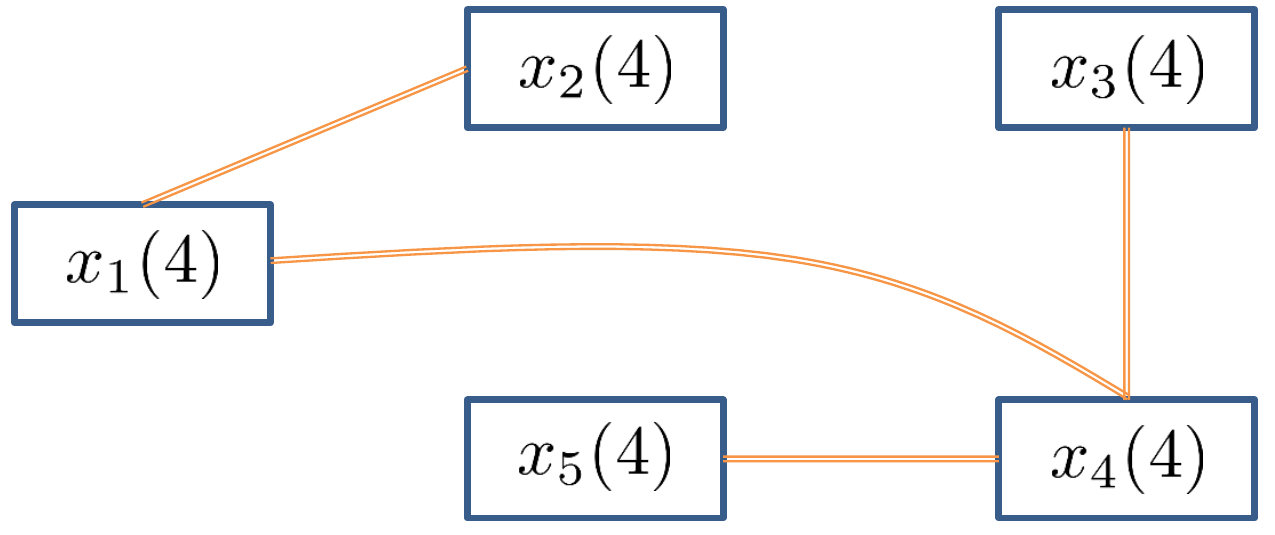}
	\caption{The graphical model $\mGM_{\betab^\sharp}$ of the node states $\xb(4)$ over $\mTG$ along the algorithm (\ref{eq:alg1}) with the edge selection order $(5,2)\prec(2,3)\prec(2,1)\prec(3,4)$.}
	\label{fig:conv_graphical_model}
\end{figure}

\medskip

\noindent{\bf Example 2.} Consider a linear equation $\Hb\yb=\zb$ with respect to $\yb\in\R^2$ where
\begin{equation}\notag
\Hb=
\begin{bmatrix}
2 & -1.5 \\
1 & 1 \\
-0.5 & 0.8
\end{bmatrix},\quad
\zb=
\begin{bmatrix}
2.5\\
3\\
-0.2
\end{bmatrix},
\end{equation}
which has an exact solution $\yb^\ast=[2\ 1]^\top$. Let $\mG_{\rm ring}$ be a $3$-node ring graph. Suppose each node $i\in\{1,2,3\}$ of $\mG_{\rm ring}$ only has the knowledge of $\Hb_i$ being the $i$-th row of $\Hb$ and $\zb_i$ being the $i$-th component of $\zb$, and aims to obtain the solution $\yb^\ast$.  Let $\mathpzc{P}_i:\R^r\to\R^r$ is the projection onto the affine space $\{\yb:\Hb_i^\top\yb=\zb_i\}$.  Then we run the following Privacy Preserving Linear Equation Solver (PPLES) introduced in Part I of this paper over $\mG_{\rm ring}$ with the initial values $y_1(0)=[-40\ \ 60]^\top, y_2(0)=[0\ \ 65]^\top, y_3(0)=[35\ \ 55]^\top$. The execution of the node states computation for the first time step is plotted in Figure \ref{fig:lae_geo}.
{Each node always has its state encrypted before sharing it, thus can keep the state, and further its knowledge of the linear equation private in the computation.}

\begin{figure}
	\centering
	\includegraphics[width=8cm]{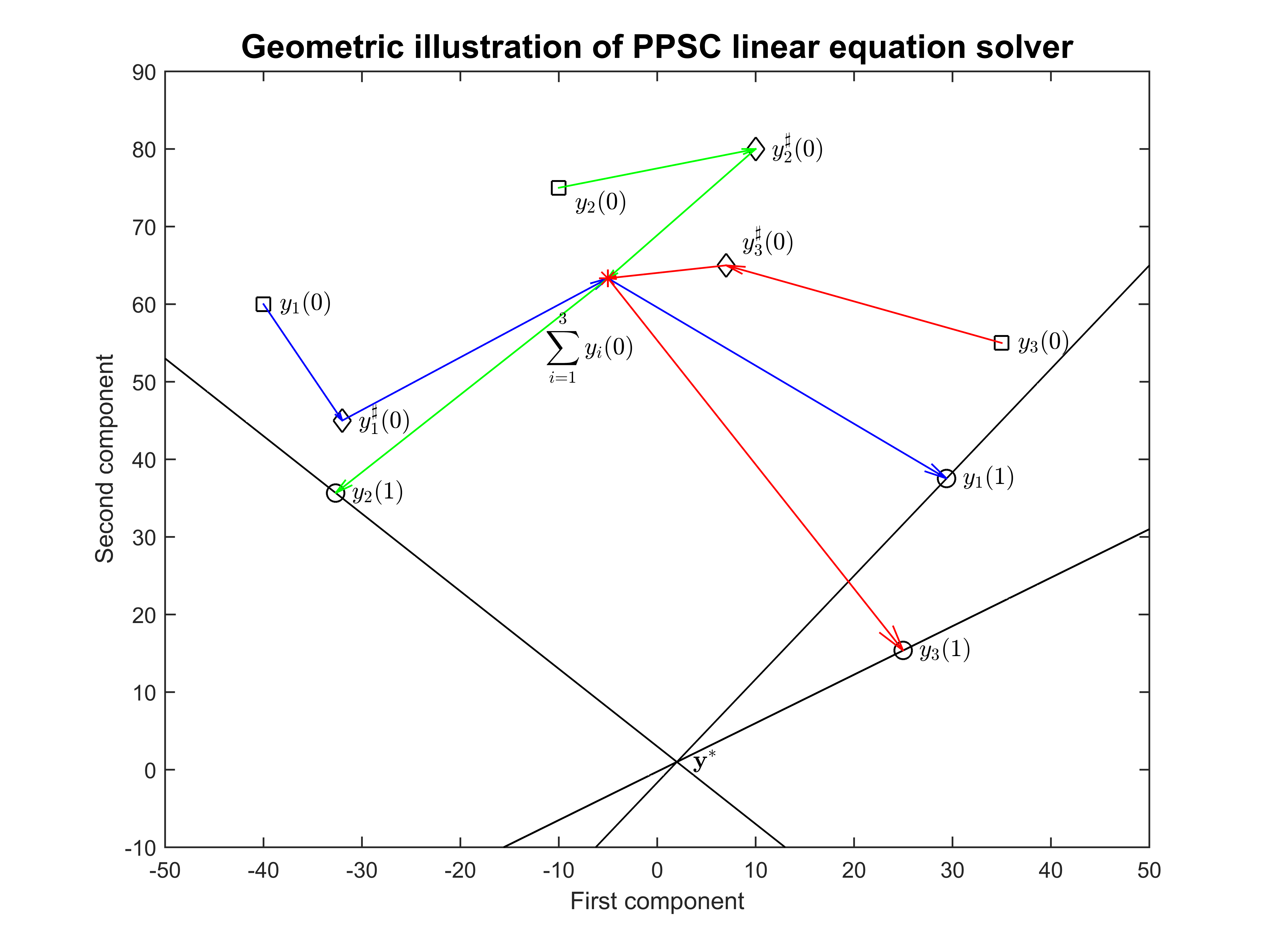}
	\caption{A geometric illustration of the privacy-preserving linear-equation solver at the first time step. The arrows in blue, green and red show the state trajectories of the node $1,2,3$, respectively. The node states avoided revealing the exact information about the solution spaces during the recursion.}
	\label{fig:lae_geo}
\end{figure}

\subsection{Discussions}
Several privacy-preserving algorithms have been proposed in the literature for consensus seeking based on the idea of injecting random offsets to node state updates \cite{Kefayati2007,huang2012differentially,manitara2013privacy,mo2017privacy}. With plain  random noise injection into the standard consensus algorithm, one inevitably lose accuracy in the convergence limit even in the probabilistic sense \cite{Kefayati2007}. However, one can show that the output of such type of algorithms can be differentially private while maintaining a certain degree of error \cite{huang2012differentially}. It was further shown in \cite{manitara2013privacy} that if one  carries out noise injection for a finite number of times and then removes the total offsets once and for all, one can obtain convergence at the exact network average. Nodes therefore needed to maintain additional memories of each  offset for the implementation of  the algorithm in \cite{manitara2013privacy}, and the offsets can be detected and reconstructed with the node states update trajectories. Recently, a comprehensive analysis was presented for the privacy protection abilities in noise-injection mechanism for average consensus algorithms with diminishing noises \cite{mo2017privacy}.

We remark that the use of random noise in the D-PPSC-Gossip algorithm is more than a random offset as one of the node states in the gossiping pair has been fully randomized. This leads to two advantages in terms of convergence and privacy-preservation:
\begin{itemize}
	\item[(i)] The D-PPSC-Gossip algorithm   converges in finite time where the network sum is accurately maintained at the algorithm output. Even its randomized  variations, which will be presented later, converge in finite time along every individual  sample. While mean-square convergence \cite{mo2017privacy} or small mean-square error \cite{huang2012differentially} implies that one needs to repeat a number of samples for a single  computation task to obtain practically accepted result.
	
	\item[(ii)] Noise-injection algorithms are vulnerable against {\it external eavesdroppers} who may hold the entire network structure information and the trajectories of node state updates. Such an eavesdropper is in general equivalent\footnote{There can still be some subtle difference between an external eavesdropper or a malicious node since as a malicious node, it will know its own initial value precisely which will influence the network state evolution.} to a malicious node connected to the entire network under the framework of   \cite{manitara2013privacy,mo2017privacy}, where   the entire network initials can be disclosed even with noisy state observations \cite{mo2017privacy}. By contrast, one can show that  under the D-PPSC-Gossip algorithm the network initials are not identifiable even for external eavesdroppers.

\end{itemize}
More importantly, we remark that as the $\betab^\sharp=\mathscr{P}_{\rm D-PPSC}(\betab)$ provides strong deniability for the $\beta_i$ of each node even if $\betab^\sharp$ is observed, the $\betab^\sharp$ can be {\em broadcast to the entire network}. Consequently, the computation of $\sum\limits_{i=1}^n\beta_i/n$ for average consensus from $\betab$ over the graph $\mG$ can be made beyond the links in $\mG$. This could in turn significantly accelerate the process of the nodes agreeing on $\sum\limits_{i=1}^n\beta_i/n$.

\subsection{MLE and MAP Strategies}
We now discuss the strategies of eavesdroppers holding the output of the D-PPSC-Gossip and try to infer information about the input. 
{Assuming $\gammab_t\overset{\text{i.i.d.}}{\sim} \normdist(0, \sigma^2)$ for $t=1,2,\dots,n-1 $, we now propose the strategy of estimating $\betab$ from $\betab^\sharp$ for the D-PPSC-Gossip algorithm based on maximum likelihood estimation (MLE) and maximum a posteriori (MAP) methods.

	\subsubsection{Maximum Likelihood Estimation}

	Let the D-PPSC-Gossip algorithm be repetitively applied on a network independently for $l>0$ times starting from the same initial value $\betab$. Let $\yb_1,\dots,\yb_l$ denote the realization of $\mathscr{P}_{\rm D-PPSC}(\betab)$.
	Let $\db_i^\top$ denote the $i$-th row of $\Db$. Then $\Db\gammab\sim\normdist(0,\Lambdab_{\Db})$ ,where
	\begin{equation}\notag
	[\Lambdab_{\Db}]_{ij}=\left\{
	\begin{aligned}
	& \sigma^2\|\db_i\|^2 & \textnormal{ if }i=j;\\
	& -\sigma^2 & \textnormal{ if }\exists \kappa\Rightarrow[\Db]_{i\kappa}[\Db]_{j\kappa}=-1;\\
	& 0 & \textnormal{ otherwise}.
	\end{aligned}
	\right.
	\end{equation}
	The distributions of $\yb_1,\dots,\yb_l$ are degenerate because $\Lambdab_{\Db}$ is singular.
	Let $\mu_1,\lambda_{11}\in\R,\ \mub_{\yb},\lambdab_{1\yb}\in\R^{n-1},\ \Lambdab_{\yb\yb}\in\R^{(n-1)\times (n-1)}$ satisfy
	\begin{equation}\label{eq:decomp1}
	\begin{aligned}
	\Cb\betab=
	\begin{bmatrix}
	\mu_1\\
	\mub_{\yb}
	\end{bmatrix}
	,\quad
	\Lambdab_{\Db}=
	\begin{bmatrix}
	\lambda_{11} & \lambdab_{1\yb}^\top\\
	\lambdab_{1\yb} & \Lambdab_{\yb\yb}
	\end{bmatrix}.
	\end{aligned}
	\end{equation}
	Here $\Lambdab_{\yb\yb}$ is noted to be nonsingular in view of the Cauchy's interlacing theorem \cite{horn1990matrix} noticing the fact that $\Lambdab_{\Db}$ is the Laplacian of a connected graph. Let $\betab_{-1}^\sharp=[\beta_2^\sharp\ \dots\ \beta^\sharp_n]^\top$. From (\ref{eq:linear_final_initial}), we know $\beta^\sharp_1\mid\betab^\sharp_{-1}$ and $\betab^\sharp_{-1}$ are both normally distributed according to \cite{rasmussen2004gaussian}:
	\begin{align}\label{eq:x_1_mid_y}
	\beta^\sharp_1\mid\betab^\sharp_{-1}&\sim\normdist(\mu_1+\phib_{1\yb}^\top(\betab^\sharp_{-1}-\mub_{\yb}),\psi_{1\yb})\\
	\betab^\sharp_{-1}&\sim\normdist(\mub_{\yb},\Lambdab_{\yb\yb}).\label{eq:y_gpr}
	\end{align}
	where
	\begin{equation}\notag
	\begin{aligned}
	\phib_{1\yb}&=\Lambdab_{\yb\yb}^{-1}\lambdab_{1\yb}\\
	\psi_{1\yb}&=\lambda_{11}-\lambdab_{1\yb}^\top\Lambdab_{\yb\yb}^{-1}\lambdab_{1\yb}.
	\end{aligned}
	\end{equation}
	{Further by $\Lambdab_{\Db}\1=0$, one has
		\begin{equation}\label{eq:phi_psi_1_0}
		\begin{aligned}
		\phib_{1\yb}=-\1,\
		\psi_{1\yb}=0.
		\end{aligned}
		\end{equation}
		Then (\ref{eq:x_1_mid_y}) and (\ref{eq:phi_psi_1_0}) imply
		\begin{equation}\label{eq:x_1_mid_y_delta}
		\beta^\sharp_1\mid\betab^\sharp_{-1} \sim \delta(\1^\top \betab^\sharp-\1^\top\betab).
		\end{equation}
		Let $\betamlefs$ be the MLE of $\betab$ from $\yb_1,\dots,\yb_l$. Based on (\ref{eq:decomp1}), (\ref{eq:y_gpr}) and (\ref{eq:x_1_mid_y_delta}), one can find by the definition of MLE
		\begin{align}
		\qquad\betamlefs\notag
		&\in\argmax\limits_{\betab\in\R^n} \prod\limits_{k=1}^l {\rm PDF}(\yb_k;\betab)\notag\\
		&= \argmax\limits_{\betab\in\R^n} \prod\limits_{k=1}^l\delta(\1^\top \yb_k-\1^\top\betab)\exp\big(-\frac{1}{2}(\yb_k-\Cb\betab)^\top\Lambdab_\ast (\yb_k-\Cb\betab)\big),\label{eq:mlepfs_objectivefunction}
		\end{align}
		where
		\begin{equation}\notag
		\begin{aligned}
		\Lambdab_\ast =
		\begin{bmatrix}
		0 & 0\\
		0 & \Lambda_{\yb\yb}^{-1}
		\end{bmatrix}.
		\end{aligned}
		\end{equation}
		Evidently the unconstrained optimization problem (\ref{eq:mlepfs_objectivefunction}) is equivalent to the constrained one
		\begin{equation} \label{eq:mle_constrained1}
		\begin{aligned}
		\min_{\betab\in\R^n}\qquad &  \sum\limits_{k=1}^l(\yb_k-\Cb\betab)^\top\Lambdab_\ast (\yb_k-\Cb\betab) \\
		{\rm s.t.} \qquad  &   \1^\top\betab=\1^\top\yb_1.
		\end{aligned}
		\end{equation}
		It is clear that Slater's condition holds for (\ref{eq:mle_constrained1}) because the only constraint is an affine equality. Introduce a Lagrange multiplier $\lambda\in\R$. Then the set of $\betamlefs$ is given by the following two KKT conditions \cite{boyd2004}:
		\begin{equation}\notag
		\begin{aligned}
		0 &= \sum\limits_{k=1}^l (\Cb^\top\Lambdab_\ast \Cb\betamlefs-\Cb^\top\Lambdab_\ast \yb_k)+\lambda\1,\\
		0 &= \1^\top(\betamlefs-\yb_1),
		\end{aligned}
		\end{equation}
		which lead to
		\begin{equation}\notag
		\begin{aligned}
		\betamlefs&\in\{\betab\in\R^n:\emlefs
		\begin{bmatrix}
		\betab\\
		\lambda
		\end{bmatrix}=
		\bmlefs,\ \lambda\in\R\},
		\end{aligned}
		\end{equation}
		where
		\begin{equation}\notag
		\emlefs=
		\begin{bmatrix}
		l\Cb^\top\Lambdab_\ast \Cb & \1\\
		\1^\top & 0
		\end{bmatrix},\quad
		\bmlefs=
		\begin{bmatrix}
		\Cb^\top\Lambdab_\ast \sum\limits_{k=1}^l\yb_k\\
		\1^\top\yb_1
		\end{bmatrix}.
		\end{equation}
	
	}

	Since $\mathbf{C}$ is singular, $\emlefs$ is not invertible, and the MLE approach is non-unique. This is true even with small (or zero) variance on  $\gammab$.

	\subsubsection{Maximum a posteriori Estimation}
	The eavesdropper can also implement a MAP estimation strategy. For MAP estimation, we let the eavesdroppers have a normal distribution assumption of $\betab$. Assume $\betab\sim\normdist(\mub_{\betab},\Lambdab_{\betab})$ with $\mub_{\betab}\in\R^n$ and $\Lambdab_{\betab}\in\R^{n\times n}$ being an $n$-by-$n$ positive definite matrix. Let $\betamapfs$ be the MAP estimators of $\betab$ upon observing $l$ sample points of $\betab^\sharp$, respectively.
	According to Bayes' rule, there holds
	\begin{equation}\label{eq:bayes_rule}
	{\rm PDF}(\betab\mid\betab^\sharp)=\frac{{\rm PDF}(\betab){\rm PDF}(\betab^\sharp\mid\betab)}{{\rm PDF}(\betab^\sharp)}.
	\end{equation}
	By (\ref{eq:bayes_rule}) and the definition of MAP, the MAP estimator of $\betab$ is solution to
	{
		\begin{equation} \label{eq:map_opt}
		\begin{aligned}
		\min_{\betab\in\R^n}\qquad &  \sum\limits_{k=1}^l(\yb_k-\Cb\betab)^\top\Lambdab_\ast (\yb_k-\Cb\betab) +(\betab-\mub_{\betab})^\top\Lambdab_{\betab}^{-1}(\betab-\mub_{\betab})\\
		{\rm s.t.} \qquad  &   \1^\top\betab=\1^\top\yb_1.
		\end{aligned}
		\end{equation}
		As a result of applying KKT conditions to the optimization problem (\ref{eq:map_opt}), the MAP estimator $\betamapfs$ can be uniquely determined by 
		\begin{equation}\notag
		\begin{aligned}
		\emapfs
		\begin{bmatrix}
		\betamapfs\\
		\lambda
		\end{bmatrix}=
		\bmapfs,\ \lambda\in\R\textnormal{ is a Lagrange multiplier},
		\end{aligned}
		\end{equation}
		where
		\begin{equation}\notag
		\begin{aligned}
		\emapfs=
		\begin{bmatrix}
		l\Cb^\top\Lambdab_\ast \Cb+\Lambdab_{\betab}^{-1} & \1\\
		\1^\top & 0
		\end{bmatrix},\ 
		\bmapfs=
		\begin{bmatrix}
		\Cb^\top\Lambdab_\ast \sum\limits_{k=1}^l\yb_k+\Lambdab_{\betab}^{-1}\mub_{\betab}\\
		\1^\top\yb_1
		\end{bmatrix}.
		\end{aligned}
		\end{equation}
	}
	
	The MAP solution is always unique, however this is solely due to the presence of the prior. Therefore, if a good prior estimate is available, MAP may improve the quality of this estimate in certain directions in the space of $\betab$, but there are always some directions for which knowledge of $\betab^\sharp$ provides no information.
(Note that this privacy analysis is based entirely on the properties of $\mathbf{C}$, and the expressions for variance would also depend on the distribution of $\gammab$.)

}

\section{Randomized PPSC Gossiping}\label{sec:rand}

In this section, we propose a randomized PPSC algorithm based on classical random gossiping \cite{boyd2006randomized}. To this end, let $\mathbf{P}$ be a stochastic matrix \cite{horn1990matrix}, i.e., a matrix with non-negative entries possessing a sum one along each row. The matrix  $\mathbf{P}$ complies with the structure of the graph $\mathrm{G}$ in the sense that $[\mathbf{P}]_{ij}>0$ if and only if $\{i,j\}\in \mathrm{E}$. Compared to deterministic gossiping, randomized gossip algorithms are ideal solutions for distributed systems, where nodes are  self-organized with asynchronous clocks \cite{boyd2006randomized}. Then we propose the following R-PPSC-Gossip algorithm.

\begin{algorithm}[H]
	{$\mathbf{Randomized\textnormal{ }PPSC\textnormal{ }Gossip\  Algorithm}$}\\
	Set $x_i(0)\gets\beta_i$ for each $i\in\mV$. For $t=1,2,\dots$, repeat the following steps.
	\begin{algorithmic}[1]
		\STATE Independently a node $i$ is selected with probability $1/n$, and then this node $i$ randomly selects a neighbor $j\in\mathrm{N}_i$ with probability $[\mathbf{P}]_{ij}$.
		\STATE Node $i$ randomly generates $\gammab_t\in\R$ and sends $x_{i}(t-1)-\gammab_t$ to $j$.
		\STATE Nodes $i,j$ update their states by
		\begin{equation}\notag
		\begin{aligned}
		x_{i}(t) &= \gammab_t\\
		x_{j}(t) &= x_{i}(t-1) + x_{j}(t-1) - \gammab_t.
		\end{aligned}
		\end{equation}
		\STATE The states remain unchanged for the rest of the nodes in $\mV$, i.e.,
		$
		x_k(t)=x_k(t-1),\ k\in\mV\setminus\{i,j\}.
		$
	\end{algorithmic}
\end{algorithm}

For any $t\ge0$, there exist random matrices $\Cb_t\in\R^{n\times n}$ and $\Db\in\R^{n\times (n-1)}$ such that $\betab^\sharp(t)=\Cb_t\betab+\Db_t\gammab$. It can be verified that any realization of $\Cb_t$ will not be of full rank, and therefore the non--identifiability condition in Theorem \ref{thm0} continues to hold. Furthermore, the value of $\Cb_t$ and $\Db_t$ will be harder to obtain, compared to the $\Cb$ and $\Db$ in (\ref{eq:linear_final_initial}), i.e., $\mathscr{P}_{\rm R-PPSC}(\betab)$ provides improved privacy preservation over $\mathscr{P}_{\rm D-PPSC}(\betab)$. Also, following a similar analysis of Theorem \ref{thmdp}, $\mathscr{P}_{\rm R-PPSC}(\betab)$ can be $\epsilon$--differentially private for sufficiently large $t$.

\subsection{Convergence Limits}
We let
$\qb_1$ denote the unique left Perron vector of $\Pb$ with $\1^\top\qb_1=1$. Let $(\qb_2,\qb_2^\prime),\dots,(\qb_n,\qb_n^\prime)$ denote the left-right generalized eigenvector pairs of $\Pb$ satisfying $\qb_i^\top\qb_i^\prime=1$ for all $i=2,\dots,n$.

\medskip

\begin{proposition}\label{prop:mean_limit}
	Let $\mu_{\gammab}$ denote the expected value of $\gammab_t$. Then along the  R-PPSC-Gossip algorithm, there holds $\lim\limits_{t\to\infty}\mathbb{E}(\xb(t)) = \qb_1\1^\top\xb(0) + \mu_{\gammab}\sum\limits_{i=2}^n\qb_i\qb_i^{\prime\top}\1$.
\end{proposition}
\begin{proof}
	In a compact form, the update of network node states can be written as
	\begin{equation}\label{eq:rand_compact_form}
	\xb(t)=\mathpzc{A}(t)\xb(t-1) + \gammab_t\mathpzc{v}(t),\ t=1,2,\dots,
	\end{equation}
	where  $\mathpzc{A}(t)$  is a random matrix and $\mathpzc{v}(t)$ is a random vector.  By the structure of the proposed algorithm, one easily know
	\begin{align}
	\mathbb{E}(\mathpzc{A}(t)) &= \Ib + \frac{1}{n}(\Pb^\top-\Ib),\label{eq:rand_Abar}\\
	\mathbb{E}(\mathpzc{v}(t))  &= \frac{1}{n}(\Ib-\Pb^\top)\1. \label{eq:rand_vbar}
	\end{align}
	
	From (\ref{eq:rand_Abar}), it is worth noting that $\mathbb{E}(\mathpzc{A}(t))$ is a primitive stochastic matrix with left Perron vector being $\qb_1$. Hence \cite{horn1990matrix}
	\begin{equation}\label{eq:Abar_conv}
	\lim\limits_{t\to\infty}\big(\mathbb{E}(\mathpzc{A}(t))\big)^t=\qb_1\1^\top.
	\end{equation}
	By performing Jordan decomposition on $\Pb^\top$, one can easily obtain from (\ref{eq:rand_Abar}) and (\ref{eq:rand_vbar})
	\begin{equation}\label{eq:series_conv}
	\lim\limits_{t\to\infty}\sum\limits_{i=0}^{t-1}\big(\mathbb{E}(\mathpzc{A}(t))\big)^i \mathbb{E}(\mathpzc{v}(t))=\sum\limits_{i=2}^n\qb_i\qb_i^\top\1.
	\end{equation}
	Then the desired  conclusion can be obtained noticing the independence of the node updates.
\end{proof}

\subsection{Convergence Rate}
Introduce $\mathcal{Q}_t$ as the event that all nodes have altered their states at least once during the time $s\in[0,t]$. Then $\mathcal{Q}_t$ holding true implies that the entire network states have been encrypted by the R-PPSC-Gossip  algorithm. Define
\begin{equation}\notag
\xib=\frac{1}{n}(\Pb+\Pb^\top)\1.
\end{equation}
Let $\xi_i$ denote the $i$-th component of $\xib$. Let $2^{\mathrm{S}}$ denote the power set of a set $\mathrm{S}$. Recall that an independent set of a graph is a subset of the graph vertex set, in which two arbitrary nodes are not adjacent in the graph \cite{godsil2013}. Based on this, a result regarding the convergence rate is shown in the following.
\begin{proposition}\label{prop:tighter_bound}
	Consider an undirected and connected graph $\mG=(\mV,\mE)$. Let the node set $\mV$ be partitioned into $\kappa > 1$ mutually disjoint independent vertex sets $\mW_1,\dots,\mW_\kappa$, which satisfy $\bigcup\limits_{i=1}^\kappa \mW_i=\mV$. Let $\pi_1,\dots,\pi_\kappa$ denote some arbitrary elements in $\mW_1,\dots,\mW_\kappa$, respectively. Then there holds
	\begin{equation}\notag
	\begin{aligned}
	\mathbb{P}(\mathcal{Q}_t)&\ge 1-\kappa+\sum\limits_{i=1}^\kappa\sum\limits_{\mU\in 2^{\mW_i\setminus\{\pi_i\}}}(-1)^{\left|\mU\right|}\big((1-\sum\limits_{j\in\mU}\xi_j)^t\notag\\
	&\quad
	-(1-\xi_{\pi_i}-\sum\limits_{j\in\mU}\xi_j)^t\big).
	\end{aligned}
	\end{equation}
\end{proposition}
\begin{proof}
	Now we define $n$ event sequences
	\begin{equation}\notag
	\{\mathcal{F}^{t}_1\}_{t=1,2,\dots},\ \dots\ ,\{\mathcal{F}^{t}_n\}_{t=1,2,\dots},
	\end{equation}
	where $\mathcal{F}^{t}_i$ represents the event that the state of node $i$ at time $t$ is unequal to its initial state. Clearly $\mathcal{Q}_t = \bigcap\limits_{i=1}^n \mathcal{F}_i^{t},\ t=1,2,\dots.$
	We have by the Fr{\'e}chet inequalities
	\begin{equation}\label{eq:frechet1}
	\mathbb{P}(\mathcal{Q}_t)\ge\sum\limits_{i=1}^\kappa \mathbb{P}\Big(\bigcap\limits_{j\in\mW_i}\mathcal{F}_j^{t}\Big)-(\kappa-1).
	\end{equation}
	For $i=1,\dots,\kappa$, it is clear by the Inclusion$-$exclusion theorem \cite{allenby2011}
	\begin{equation}
	\begin{aligned}
	&\quad\mathbb{P}\big(\bigcap\limits_{j\in\mW_i}\mathcal{F}_j^{t}\big)=\sum\limits_{\mU\in2^{\mW_i\setminus\{\pi_i\}}}(-1)^{\left|\mU\right|}\mathbb{P}\Big(\big(\bigcap\limits_{j\in\mU} \overline{\mathcal{F}}_j^{t}\big)\mcap \mathcal{F}_{\pi_i}^{t}\Big)\\
	&=\sum\limits_{\mU\in2^{\mW_i\setminus\{\pi_i\}}}(-1)^{\left|\mU\right|}\mathbb{P}\Big(\bigcap\limits_{j\in\mU} \overline{\mathcal{F}}_j^{t}\Big)\cdot
	\mathbb{P}\Big(\mathcal{F}_{\pi_i}^{t}\mid\bigcap\limits_{j\in\mU} \overline{\mathcal{F}}_j^{t}\Big).\label{eq:minusoneU}
	\end{aligned}
	\end{equation}
	It can be noted from the definition of $\xib$ that $\xi_i\in(0,1]$ with $i=1,\dots,n$ represents the probability of the event that node-to-node communication involves node $i$ in a time slot. For each $\mU\in2^{\mW_i\setminus\{\pi_i\}}$, it is known that any two nodes in $\mU$ are not adjacent, and thus
	\begin{align}
	\mathbb{P}\Big(\bigcap\limits_{j\in\mU} \overline{\mathcal{F}}_j^{t}\Big) = (1-\sum\limits_{j\in\mU}\xi_j)^t.\label{eq:prior1}
	\end{align}
	Similarly, there holds
	\begin{equation}\label{eq:joint1}
	\mathbb{P}\Big(\big(\bigcap\limits_{j\in\mU} \overline{\mathcal{F}}_j^{t}\big)\mcap\overline{\mathcal{F}}_{\pi_i}\Big) = (1-\xi_{\pi_i}-\sum\limits_{j\in\mU}\xi_j)^t.
	\end{equation}
	Then (\ref{eq:prior1}) and (\ref{eq:joint1}) yield
	\begin{align}
	\mathbb{P}\Big(\mathcal{F}_{\pi_i}^{t}\mid\bigcap\limits_{j\in\mU} \overline{\mathcal{F}}_j^{t}\Big)&= 1-\mathbb{P}\Big(\overline{\mathcal{F}}_{\pi_i}^{t}\mid\bigcap\limits_{j\in\mU} \overline{\mathcal{F}}_j^{t}\Big)\notag\\
	&= 1-\frac{\mathbb{P}\Big(\big(\bigcap\limits_{j\in\mU} \overline{\mathcal{F}}_j^{t}\big)\mcap\overline{\mathcal{F}}_{\pi_i}\Big)}{\mathbb{P}\Big(\bigcap\limits_{j\in\mU} \overline{\mathcal{F}}_j^{t}\Big)}\notag\\
	&= 1-(1-\frac{\xi_{\pi_i}}{1-\sum\limits_{j\in\mU}\xi_j})^t\label{eq:cond1}.
	\end{align}
	The proof is completed by (\ref{eq:frechet1}), (\ref{eq:minusoneU}), (\ref{eq:prior1}) and (\ref{eq:cond1}).
\end{proof}

We can also quantify the rate of convergence by the following $\epsilon$-encryption time.
\begin{definition} \label{def:rand_encryption_time}
	For any $\epsilon\in(0,1)$, the $\epsilon$-encryption time for an undirected and connected graph $\mG=(\mV,\mE)$ with $n$ nodes and a randomized PPSC gossiping associated with the edge selection probability given in $\Pb\in\R^{n\times n}$ is defined by
	\begin{equation}\notag
	T_\epsilon(\mG,\Pb)=\inf\{t: 1-\mathbb{P}(\mathcal{Q}_t)\le\epsilon\}.
	\end{equation}
\end{definition}
Define $\xi_{\rm m}=\min\limits_i \xi_i$. For the $\epsilon$-encryption time, we have the following proposition.

\medskip

\begin{proposition}\label{prop:encryption_time_bounds}
	For any $\epsilon\in(0,1)$, the $\epsilon$-encryption time associated with graph $\mG$ and matrix $\Pb$ for the random PPSC algorithm satisfies
	\begin{equation}\notag
	\frac{\log\epsilon}{\log(1-\xi_m)} \le T_\epsilon(\mG,\Pb) \le \frac{\log\epsilon-\log n}{\log(1-\xi_m)}.
	\end{equation}
\end{proposition}
\begin{proof}
	Let $\mathcal{F}^t_i,\ i=1,\dots,n,\ t=1,2,\dots$ be as defined in the proof of Proposition \ref{prop:tighter_bound}. Then it can be concluded
	\begin{equation}\label{eq:rand_xi_neq}
	\mathbb{P}(\mathcal{F}_i^{t}) = 1-\mathbb{P}(\overline{\mathcal{F}}_i^{t})=1-(1-\xi_i)^t.
	\end{equation}
	By (\ref{eq:rand_xi_neq}) and the Fr{\'e}chet inequalities \cite{frechet1935generalisation}, we have{
		\begin{equation}
		\begin{aligned}
		\mathbb{P}(\mathcal{Q}_t) \ge \sum\limits_{i=1}^n \mathbb{P}(\mathcal{F}_i^{t})-(n-1)= 1-\sum\limits_{i=1}^n(1-\xi_i)^t\ge 1-n(1-\xi_m)^t.\label{eq:rand_encryption_time_lowerbound}
		\end{aligned}
		\end{equation}
		By the Fr{\'e}chet inequalities \cite{frechet1935generalisation}, one also has
		\begin{align}
		\mathbb{P}(\mathcal{Q}_t) \le \min_i\{\mathbb{P}(\mathcal{F}_i^{t})\} = 1-(1-\xi_m)^t. \label{eq:rand_encryption_time_upperbound}
		\end{align}}
	Clearly (\ref{eq:rand_encryption_time_lowerbound}) and (\ref{eq:rand_encryption_time_upperbound}) complete the proof.
\end{proof}

We would like to point out that techniques from optimizing the structure of the work and the selection of the neighbors \cite{boyd-2004,doerr-2012,shiton2016} might significantly accelerate the convergence rate of the algorithm.

\subsection{Resilience vs Privacy   Trade-off}
Throughout the running of the algorithm (\ref{eq:alg1}) with deterministic or randomized edge selection over a network, a circumstance may occur that a node   drops out of the network   at a random time. Let us assume that independently at each time step, each node of the network has a probability $\pdropout>0$ of dropping out. We focus on a particular time instance $\tdropout>0$. Let $\mGrand_t=(\mVrand_t,\mErand_t)$ denote the random network at time $t=0,1,2,\dots$. It is clear that
\begin{align}
\mVrand_{t}&\subset\mVrand_{t-1},\notag\\
\mErand_{t}&=\mErand_{t-1}\setminus\big\{\{i,j\}:i\in\mVrand_{t-1}\setminus\mVrand_{t},j\in\mVrand_{t-1}\big\}.\notag
\end{align}

On one hand, from the point of view of $\Tail(\asf_{\tdropout})$,     nodes in $\mVrand_{\tdropout}$ should have their states sum only yielding   a small change  compared to that at time $\tdropout-1$ due to the node dropout, i.e.,
\begin{equation}\label{eq:expetation_drop_out}
\begin{aligned}
\mathbb{E}\Big|\sum\limits_{i\in\mVrand_{\tdropout}}x_i(\tdropout)-\sum\limits_{i\in\mVrand_{\tdropout}}x_i(\tdropout-1)\Big|^2
\end{aligned}
\end{equation}
is preferred to be as small as possible, where $\mathbb{E}(\cdot)$ is with respect to the randomness from node dropouts. By direct calculation we can see that
\begin{equation}\label{eq:zeta_omega}
\begin{aligned}
\mathbb{E}\Big|\sum\limits_{i\in\mVrand_{\tdropout}}x_i(\tdropout)-\sum\limits_{i\in\mVrand_{\tdropout}}x_i(\tdropout-1)\Big|^2= 2\pdropout(1-\pdropout)\mathbb{E} \Big|x_{\Tail(\asf_{\tdropout})}(\tdropout-1)-\gammab_{\tdropout}\Big|^2.
\end{aligned}
\end{equation}
It worth emphasizing that $\omega_{\tdropout}=x_{\Tail(\asf_{\tdropout})}(\tdropout-1)-\gammab_{\tdropout}$ is also the packet of communication during the node pair interaction at time $\tdropout$, which will be the error added into the network sum if the communication fails at the receiving node $\Head(\asf_{\tdropout})$. Therefore,  $$
\mathfrak{U}_{\rm R}=2\pdropout(1-\pdropout)\mathbb{E}(|x_{\Tail(\asf_{\tdropout})}(\tdropout-1)-\gammab_{\tdropout}|^2)
$$ serves as a natural
network resilience metric. On the other hand,  by receiving the packet $\omega_{\tdropout}$, $\Head(\asf_{\tdropout})$ or a third party can possibly recover the state $x_{\Tail(\asf_{\tdropout})}({\tdropout}-1)$. In that case,  $\Tail(\asf_{\tdropout})$ would hope the following conditional entropy (see e.g., \cite{cover2012})   $$h(x_{\Tail(\asf_{\tdropout})}({\tdropout}-1)\mid\omega_{\tdropout})$$
i.e., the entropy of $x_{\Tail(\asf_{\tdropout})}({\tdropout}-1)$ given  $\omega_{\tdropout}$,
to be as large as possible. As a result,  $$\mathfrak{U}_{\rm P}=h(x_{\Tail(\asf_t)}(t-1)\mid\omega_t)$$
can be a good privacy preservation metric for any time $t$.

With normal distribution assumptions on both the $\beta_i$s and the $\gammab_i$s, $x_{\Tail(\asf_t)}(t-1)$ is normally distributed with its mean and the variance denoted as $\tilde{\mu}$ and $\tilde{\sigma}^2$, respectively. We can now conclude that
\begin{align}
\mathfrak{U}_{\rm R} &= 2\pdropout(1-\pdropout)\big((\tilde{\mu}-\varrho_\gammab)^2 + \tilde{\sigma}^2 + \sigma_{\gammab}^2\big)\label{eq:MR_eq},\\
\mathfrak{U}_{\rm P} &= \frac{1}{2}\log 2\pi e-\frac{1}{2}\log(\frac{1}{\tilde{\sigma}^2}+\frac{1}{\sigma_\gammab^2}).\label{eq:MP_eq}
\end{align}
Thus, a tradeoff between network resilience and privacy preservation can be characterized by
\begin{equation}\label{eq:trade-off_obj}
\begin{aligned}
\argmin_{\varrho_\gammab,\sigma_\gammab^2}\qquad &  \mathfrak{U}_{\rm R} - \nu \mathfrak{U}_{\rm P},
\end{aligned}
\end{equation}
where $\nu\in\R^{+}$ is a parameter that weights the importance of the resilience and the privacy preservation capability. With (\ref{eq:MR_eq}) and (\ref{eq:MP_eq}),  (\ref{eq:trade-off_obj}) yields a unique solution
\begin{equation}\label{eq:trade-off_solution}
\begin{aligned}
\varrho_\gammab &= \tilde{\mu},\\
\sigma_\gammab^2 &= \frac{\tilde{\sigma}}{2}(\sqrt{\tilde{\sigma}^2+\frac{\nu}{\pdropout(1-\pdropout)}}-\tilde{\sigma}).
\end{aligned}
\end{equation}
The relation (\ref{eq:trade-off_solution}) provides an inspiration on how we can generalize the algorithm (\ref{eq:alg1}) to the adaptively generated noise sequence  $\gamma_{t}$.  Letting the random variable $\gammab_t$ have a state-dependent mean and variance related  to the state of the node that generates it, one can achieve a degree of a balance between resilience and privacy preservation.

\section{Conclusions}\label{sec:conclusion}
We have provided gossip-based realizations to PPSC mechanism over a network.  For deterministic node updates, we established   necessary and sufficient conditions on the dependence of two arbitrary nodes' final states,  and characterized the their pairwise dependence with stochastic graphical models. For randomized edge selection, the convergence limit and convergence rate for full-network encryption were established. It was shown that the resulting PPSC mechanism can preserve privacy in terms of input identification and differential privacy when the output is observed.  As an extension, the trade-off between resilience and privacy preservation was studied. Future work includes the design of optimal network structure for information preservation, and study of the fundamental limits between privacy preservation and computation efficiency in distributed algorithms.


%

\medskip

\medskip

\medskip

\section*{Appendices}

\medskip

\subsection*{A. Key Lemmas}
The D-PPSC-Gossip algorithm can be re-expressed by the following equations:
\begin{equation}\label{eq:alg1}
\begin{aligned}
& x_{\Tail(\asf_{t})}(t) = \gammab_{t}\\
& x_{\Head(\asf_{t})}(t) = x_{\Head(\asf_{t})}(t-1) + x_{\Tail(\asf_{t})}(t-1) - \gammab_t\\
& x_{i}(t) = x_{i}(t-1),\ i\in\mV\setminus\{\Tail(\asf_{t}),\Head(\asf_{t})\}
\end{aligned}
\end{equation}
for $t=1,2,\dots,n-1 $. 
Noting that the sum of node states remains  the same over time along the D-PPSC-Gossip algorithm, i.e., $\sum\limits_{i=1}^n x_{i}(t)=\sum\limits_{i=1}^n x_{i}(0)$ for all $0< t\le n-1 $, it is evident that each node $i$'s state can be expressed as
\begin{equation}\label{eq:c_d_def}
x_{i}(t) = \sum\limits_{j=1}^n c_{ij}(t)\beta_j + \sum\limits_{j=1}^{t}d_{ij}(t)\gammab_j,
\end{equation}
where $c_{ij}(t)\in\{0,1\}$ and $d_{ij}(t)\in\{-1,0,1\}$ represents random variable $d_{ij}(t)\gammab_j$ appears in node state $x_i(t)$. In a state $x_i(t)$, $d_{ij}(t)\gammab_j$ with $d_{ij}(t)\neq0$ is a \emph{random component}.
The computation along (\ref{eq:alg1}) is a process of the $x_i(t)$ gaining and losing these random components.
We note a few basic rules for that process.
\begin{enumerate}[(i)]
	\item For any time $t\ge s$, $\gammab_s$ and $-\gammab_s$ belong to different node states, {i.e., appear in the states of two different nodes.}
	\item Any random component can only be gained at a head from a tail along their directed link.
	\item The random components do not change their signs when being gained or lost.
\end{enumerate}
The following lemma illustrates the way that a random variable is passed from the state of one node to that of another.
\begin{lemma}\label{lem:RV_transfer}
	Suppose $n>2$. Let $\tbeg\in\{1,2,\dots,n-1 -1\}$ and $s\in\{1,2,\dots,\tbeg\}$. Let node $i$ and node $j$ satisfy $ d _{i s}(\tbeg)\neq 0$ and $ d _{j s}(\tbeg)= 0$, respectively. Denote $i_0,i_1,\dots,i_l$ as the unique undirected path in $\mT_{\mG}$ connecting node $i$ and $j$ with $i=i_0$ and $j=i_l$. Define $\tend\in\{\tbeg+1,2,\dots,n-1 \}$. Then $d_{i_l s}(\tend)= d_{i_0 s}(\tbeg)$ if and only if the following conditions hold:
	\begin{enumerate}[(i)]
		\item $i_0,i_1,\dots,i_l$ is a directed path.
		
		\item $\asf_{\tbeg}\prec(i_0,i_1)\prec\dots\prec(i_{l-1},i_l)\preceq\asf_{\tend}$.
		
		\item \begin{itemize}
			\item[a)] $(i_0,i^\ast)\prec\asf_{\tbeg}\textnormal{ or }(i_0,i_{1})\preceq(i_0,i^\ast)$ when $(i_0,i^\ast)\in\mET$;
			\item[b)] $(i_b,i^\ast)\prec(i_{b-1},i_{b})\textnormal{ or }(i_b,i_{b+1})\preceq(i_b,i^\ast)$ when $(i_b,i^\ast)\in\mET$ with $0<b<l$;
			\item[c)] $(i_l,i^\ast)\prec(i_{l-1},i_{l})\textnormal{ or }\asf_{\tend}\preceq(i_l,i^\ast)$ when $(i_l,i^\ast)\in\mET$.
		\end{itemize}
	\end{enumerate}
\end{lemma}
\noindent{\it Proof.} 
The conclusion $d_{i_l s}(\tend)= d_{i_0 s}(\tbeg)\neq0$ describes that the random component $d_{i_l s}(\tend)\gammab_{s}$ is passed from node $i_0$ to node $i_l$ during the time interval $[\tbeg,\tend]$. Before presenting the proof, we provide some intuitive explanation on the three conditions. Condition (i) confirms the orientations of edges in the path $i_0,\dots,i_l$, while the order of the edge selection is given in Condition (ii). Condition (iii) prevents the random component $d_{is}(\tbeg)\gammab_s$ from being passed to the nodes that are not in the path $i_0,\dots,i_l$. Next the sufficiency and necessity of the conditions are proved, respectively.\\
\emph{Sufficiency.}
Let $\asf_{t_k}=(i_k,i_{k+1})$ for $k=0,1,\dots,l-1$. Then Condition (ii) is equivalent to
\[
\tbeg<t_0<t_1<\dots<t_{l-1}\le\tend.
\]
$d_{is}(\tbeg)\neq0$ states that the random variable $d_{is}(\tbeg)\gammab_s$ appears in the node state $x_i(\tbeg)$. Condition (iii) a) describes that $i_0$ is not the tail of any edges $\asf_{t}$ with $\tbeg\le t < t_0$, and hence guarantees that the random components in $x_{i_0}(\tbeg)$, including $d_{i_0s}(\tbeg)\gammab_s$, will be kept in $x_{i_0}(t)$ for all $\tbeg\le t < t_0$.
At time $t_0$, node $i_0$ and node $i_1$ are chosen to mutually communicate and get their states updated according to the algorithm (\ref{eq:alg1}).
As a result, $x_{i_0}(t_0)$ loses the random component $d_{i_0s}(\tbeg)\gammab_s$, which now becomes a component of $x_{i_1}(t_1)$. The similar analysis applied at time $t_1,\dots,t_{l-1}$ recursively shows the random component $d_{i_0s}(\tbeg)\gammab_s$ appears in $x_{i_l}(t_{l-1})$ under Condition (i), (ii) and (iii) b).
Finally, Condition (iii) c) tells that random component $d_{is}(\tbeg)\gammab_s$ will be kept in $x_{i_l}(t)$ for $t_{l-1}\le t\le \tend$,
which completes the proof of sufficiency.\\
\medskip
\emph{Necessity.}
In the following, we prove the necessity of the three conditions. Assume the random component $d_{i_0 s}(\tend)\gammab_{s}$ is passed from node $i_0$ to node $i_l$ during the time interval $[\tbeg,\tend]$. Recall that the transition of the random variable $d_{i_0 s}(\tend)\gammab_{s}$ occurs from the tail of an edge to the head at each time step. Thus one has $(i_0,i_1),\dots,(i_{l-1},i_l)\in\mET$ due to the uniqueness of the undirected path $i_0,i_1,\dots,i_l$ in $\mT_\mG$, and completes the proof of the necessity of Condition (i).\\
According to the orientations of $(i_k,i_{k+1})$, there exist time $\tau_k$ when $d_{i_0 s}(\tend)\gamma_{s}$ appears in the node state $x_{i_k}(\tau_k)$ for $k=0,1,\dots,l$. By the uniqueness of the directed path from $i_0$ to $i_l$, $\tau_k$s can be arranged in the following order
\begin{equation}\label{eq:path_tau_order}
\tau_0<\tau_1<\dots<\tau_l.
\end{equation}
The definition of $\asf_{t_k}$ directly implies that for $k=0,1,\dots,l-1$
\begin{equation}\label{eq:t_k_interval}
\tau_k<t_k\le\tau_{k+1}.
\end{equation}
It is clear (\ref{eq:path_tau_order}) and (\ref{eq:t_k_interval}) yield
\begin{equation}\label{eq:t_k_order}
t_0<t_1<\dots<t_{l-1}.
\end{equation}
It is evident that if $\asf_{t_0}\le\asf_{\tbeg}$, the random component $d_{i_0 s}(\tend)\gamma_{s}$ can never appear in the state of node $i_1$, and thereby the states of $i_2,\dots,i_l$. Thus one has $\asf_{\tbeg}\prec\asf_{t_0}$.
In addition, node $i_l$ cannot gain the random component $d_{i_0 s}(\tend)\gammab_{s}$ when the time $t<t_{l-1}$, which implies $\asf_{t_{l-1}}\preceq\asf_{\tend}$.
Then the necessity of Condition (ii) can be seen from (\ref{eq:t_k_order}), $\asf_{\tbeg}\prec\asf_{t_0}$ and $\asf_{t_{l-1}}\preceq\asf_{\tend}$.\\
Suppose Condition (iii) a) does not hold for contradiction. Then there exists a node set $\mathrm{I}^\ast\subset\mV\setminus\{i_0,\dots,i_l\}$ such that $(i_0,i^\ast)\in\mTG$ and $(i_0,i^\ast)\prec\asf_{t_0}$ for all $i^\ast\in\mathrm{I}^\ast$. As a result, the random component $d_{i_0 s}(\tend)\gamma_{s}$ will be passed to some node in $\mathrm{I}^\ast$, which is impossible to be passed to $i_l$ again. Thus Condition (iii) a) must hold. Similarly, due to the uniqueness of the path between two arbitrary nodes in spanning trees, the random component $d_{i_0 s}(\tend)\gamma_{s}$ must be always held along path $i_0,i_1,\dots,i_l$ during time $\tbeg\le t\le\tend$, which proves the necessity of Condition (iii).
\hfill$\square$

\begin{lemma}\label{lem:final_node_states}
	The terminal state $\betab^\sharp=[\beta^\sharp\ \dots\ \beta^\sharp_n]^\top$ has the following properties.
	\begin{enumerate}[(i)]
		\item Let $s\in\{1,2,\dots,n-1 \}$. Then there exists unique $i,j$ with $i\neq j$ such that $ d _{is}(n-1 )=1$ and $ d _{js}(n-1 )=-1$.
		\item If $\beta^\sharp_i$ and $\beta^\sharp_j$ are dependent, then there exists a unique $s\in\{1,2,\dots,n-1 \}$ that satisfy
		\begin{equation}\notag
		d _{is}(n-1 )d _{js}(n-1 )=-1.
		\end{equation}
	\end{enumerate}
\end{lemma}
\noindent{\it Proof.} 
(i) For the endpoints of an arbitrary edge $\asf_{s}$ with $s\in\{1,2,\dots,n-1 \}$, the algorithm (\ref{eq:alg1}) yields at time $\hat{t}$
\begin{equation}\label{eq:new_rv_pairs}
\begin{aligned}
x_{\Tail(\asf_{s})}(\hat{t}) &= \gammab_{s}\\
x_{\Head(\asf_{s})}(\hat{t}) &= x_{\Tail(\asf_{s})}(s-1) + x_{\Head(\asf_{s})}(s-1) - \gammab_{s}.
\end{aligned}
\end{equation}
From (\ref{eq:new_rv_pairs}), we see that a pair of random variables $\gammab_{s},-\gammab_{s}$ are added to the states of $\Tail(\asf_{s})$ and $\Head(\asf_{s})$ respectively at time $s$. Next we analyze the random variable $\gammab_{s}$ in $x_{\Tail(\asf_{s})}(t)$ for $t>s$. If $\Tail(\asf_{s})\neq\Tail(\asf_t)$ and $\Tail(\asf_{s})\neq\Head(\asf_t)$ for all $t>s$, then it is clear that $\gammab_{s}$ remains in $x_{\Tail(\asf_{s})}(t)$ for all $t>s$. For any time $t>s$ with $\Tail(\asf_{s})=\Head(\asf_t)$, the random variables held by $x_{\Tail(\asf_{s})}(t-1)$ are still kept in $x_{\Tail(\asf_{s})}(t)$. Therefore, the only way that the state of node $x_{\Tail(\asf_{s})}$ loses $\gammab_{s}$ is to let node $\Tail(\asf_{s})$ be the tail of edge $\asf_t$ for some $t>s$. Suppose there exists a nonempty set $\mathcal{T}$ such that $\Tail(\asf_{s})=\Tail(\asf_t)$ for all $t\in\mathcal{T}$. Define $\bar{t}=\min\{t:t\in\mathcal{T}\}$. Then by the algorithm (\ref{eq:alg1})
\begin{equation}\notag
\begin{aligned}
x_{\Tail(\asf_{s})}(\bar{t}) &= \gammab_{\bar{t}}\\
x_{\Head(\asf_{s})}(\bar{t}) &= x_{\Tail(\asf_{s})}(\bar{t}-1) + x_{\Head(\asf_{\bar{t}})}(\bar{t}-1) - \gammab_{\bar{t}},
\end{aligned}
\end{equation}
from which we see that the random variable $\gammab_{s}$ transfers to $x_{\Head(\asf_{\bar{t}})}(\bar{t})$ without changing its sign. It can be concluded by applying the same analysis that $\gammab_s$ exists in one and only one of all node states for all $t\ge s$. Analogously, we can easily know that $-\gammab_{\hat{t}}$ has the same properties. This completes the proof of (i).\\
(ii) Since $\xb_i(n-1 )$ and $\xb_j(n-1 )$ are dependent, there exists a set $\mathcal{S}\subset\{1,2,\dots,n-1 \}$ such that $d_{is}(n-1 )\gammab_s$ appears in $\xb_i(n-1 )$ and $d_{js}(n-1 )\gammab$ appears in $\xb_j(n-1 )$ for all $s\in\mathcal{S}$. Let $i_0,i_1,\dots,i_l$ denote the unique undirected path in $\mT_\mG$ that connects node $i$ and node $j$ with $i_0=i$ and $i_l=j$. Thus we only need to prove $|\mathcal{S}|=1$. For every $s\in\mathcal{S}$, under the algorithm (\ref{eq:alg1}), $\gammab_s$ is held by $x_{\Tail(\asf_s)}(s)$ and $-\gammab_s$ is held by $x_{\Head(\asf_s)}(s)$. Without loss of generality, we assume $x_{i_0}(n-1 )$ holds $\gammab_s$ and $x_{i_l}(n-1 )$ holds $-\gammab_s$. If $s=n-1 $, then it is necessary $\Tail(\asf_s)=i_0$ and $\Head(\asf_s)=i_l$. If $s\neq n-1 $, at most one of $\Tail(\asf_s)=i_0$ and $\Head(\asf_s)=i_l$ holds. When $\Tail(\asf_s)=i_0$ and $\Head(\asf_s)\neq i_l$, $-\gammab_s$ has to transfer from to $x_{\Head(\asf_s)}(s)$ from $x_{i_l}(n-1 )$. By Lemma \ref{lem:RV_transfer}, the process of transfer requires the path that starts from $\Head(\asf_s)$ and ends at $i_l$ is a directed path. In this case, $i_0,i_1,\dots,i_l$ is a directed path and $\asf_s=(i_0,i_1)$. Similarly, $i_l,i_{l-1},\dots,i_0$ is a directed path and $\asf_s=(i_l,i_{l-1})$ when $\Tail(\asf_s)\neq i_0$ and $\Head(\asf_s)= i_l$. In a general case in which $\Tail(\asf_s)\neq i_0,\Head(\asf_s)\neq i_l$, the paths that connect $\Tail(\asf_s)$ and $i_0$, $\Head(\asf_s)$ and $i_l$ are both directed paths by Lemma \ref{lem:RV_transfer}, leading to that $\asf_s$ is in $i_0,i_1,\dots,i_l$. In conclusion, $\asf_s$ must be in path $i_0,i_1,\dots,i_l$ and the cases studied above can be summarized as follows.
\begin{enumerate}[(i)]
	\item $\Tail(\asf_s)=i_0,\Head(\asf_s)=i_l\Rightarrow \asf_s=(i_0,i_l)$ is an edge in $\mTG\Rightarrow$ $\asf_s$ is unique;
	\item $\Tail(\asf_s)=i_0,\Head(\asf_s)\neq i_l\Rightarrow i_0,i_1,\dots,i_l$ is directed path and $\asf_s=(i_0,i_1)\Rightarrow$ $\asf_s$ is unique;
	\item $\Tail(\asf_s)\neq i_0,\Head(\asf_s)= i_l\Rightarrow i_l,i_{l-1},\dots,i_0$ is directed path and $\asf_s=(i_l,i_{l-1})\Rightarrow$ $\asf_s$ is unique;
	\item $\Tail(\asf_s)\neq i_0,\Head(\asf_s)\neq i_l\Rightarrow $ the paths from $\Tail(\asf_s)$ to $i_0$ and from $\Head(\asf_s)$ to $i_l$ are directed paths and $\asf_s=(i_l,i_{l-1})\Rightarrow$ $\asf_s$ is unique.
\end{enumerate}
This completes the proof of (ii). \hfill$\square$

\subsection*{B. Proof of Theorem \ref{thm0}}

Note that the matrix $\Cb$ in (\ref{eq:linear_final_initial}) has its $\Tail(\asf_{n-1 })$--th row being zero. Thus $\rank(\Cb)<n$, which implies all $\etab\in\ker(\Cb)$ satisfy
\begin{equation}\notag
{\rm PDF}(\betab^\sharp\mid\betab)={\rm PDF}(\betab^\sharp\mid\betab+\etab).
\end{equation}
This directly shows that $\betab$ is non-identifiable and completes the proof.

\subsection*{C. Proof of Theorem \ref{thmdp}}

Recall that the output mechanism of the D-PPSC-Gossip algorithm can be described by (\ref{eq:linear_final_initial}).
Let $\mG[\mathrm{S}]$ with $\mathrm{S}\subset\mV$ be the graph whose node set is $\mV$ and edge set is the subset of edges in $\mG$ with both endpoints in $\mV$. For the matrix $\Db\in\R^{n\times(n-1)}$, we provide the following lemma.
\begin{lemma}\label{lem:D}
	$\Db\in\R^{n\times(n-1)}$ has full column rank.
\end{lemma}
\begin{proof}
	Let $\mathbf{d}_1,\dots,\mathbf{d}_{n-1}$ be the columns of $\Db$ from left to right. We now show that $\mathbf{d}_1,\dots,\mathbf{d}_{n-1}$ are linearly independent. Consider the equation with respect to $a_1,\dots,a_{n-1}\in\R$:
	\begin{equation}\label{eq:ind_1}
	\sum\limits_{i=1}^{n-1}a_i\mathbf{d}_i = 0.
	\end{equation}
	Let the nodes in the graphical model $\mGM_{\betab^\sharp}$ with degree one form a set $\mV_1\subset\mV$. For each $j\in\mV_1$, by Lemma \ref{lem:final_node_states}, $\Db$ has a $1$ and a $-1$ in the $j$--th column, denoted as $[\Db]_{i^1_jj}=1$ and $[\Db]_{i^2_jj}=-1$, respectively, while all the other elements in that column are zeros. Since the nodes $j\in\mV_1$ have degree one, $[\Db]_{i_jj}$ with $i_j=i^1_j$ or $i_j=i^2_j$ is the only nonzero element in the $j$--th row of $\Db$. Therefore, (\ref{eq:ind_1}) admits $a_{i_j}=0$ for all $j\in\mV_1$ and (\ref{eq:ind_1}) gives the following equation with respect to $a_{i_j},j\in\mV\setminus\mV_1$:
	\begin{equation}\label{eq:ind_2}
	\sum\limits_{i=1,i\neq i_j,j\in\mV_1}^{n-1} a_i\mathbf{d}_i = 0.
	\end{equation}
	Evidently, $\mGM_{\betab^\sharp}[\mV\setminus\mV_1]$ is still a tree. We let $\mV_2$ denote the set of nodes with degree one in $\mGM_{\betab^\sharp}[\mV\setminus\mV_1]$. Then we can continue to repeat the process from (\ref{eq:ind_1}) to (\ref{eq:ind_2}), and finally obtain $a_1=\dots=a_{n-1}=0$, which shows that $\mathbf{d}_1,\dots,\mathbf{d}_{n-1}$ are linearly independent and completes the proof.
	%
\end{proof}

We are now ready to prove the desired theorem. 	Consider two adjacent initial conditions $\betab,\betab^\prime\in\R^n$. For any $\1^\top\betab^\sharp=\1^\top\betab=\1^\top\betab^\prime$, there holds
\begin{align}\label{eq:dp_1}
\frac{\Pr\big(\mathscr{P}_{\rm D-PPSC}(\betab)=\betab^\sharp\big)}{\Pr\big(\mathscr{P}_{\rm D-PPSC}(\betab^\prime)=\betab^\sharp\big)}=\frac{\Pr(\Cb\betab+\Db\gammab=\betab^\sharp)}{\Pr(\Cb\betab^\prime+\Db\gammab=\betab^\sharp)}.
\end{align}
By Lemma \ref{lem:D}, $(\Db^\top\Db)^{-1}\Db^\top$ is the left inverse $\Db$. By (\ref{eq:dp_1}), we have
\begin{align}\label{eq:dp_2}
\frac{\Pr\big(\mathscr{P}_{\rm D-PPSC}(\betab)=\betab^\sharp\big)}{\Pr\big(\mathscr{P}_{\rm D-PPSC}(\betab^\prime)=\betab^\sharp\big)}=\frac{{\rm PDF}\big((\Db^\top\Db)^{-1}\Db^\top(\betab^\sharp-\Cb\betab)\big)}{{\rm PDF}\big((\Db^\top\Db)^{-1}\Db^\top(\betab^\sharp-\Cb\betab^\prime)\big)}.
\end{align}
In (\ref{eq:dp_2}), if $\betab-\betab^\prime\in\ker(\Cb)$, then $\frac{\Pr\big(\mathscr{P}_{\rm D-PPSC}(\betab)=\betab^\sharp\big)}{\Pr\big(\mathscr{P}_{\rm D-PPSC}(\betab^\prime)=\betab^\sharp\big)}=1<e^\epsilon$ for any $\epsilon>0$, which preserves differential privacy. Next we suppose $\betab-\betab^\prime\notin\ker(\Cb)$. By (\ref{eq:dp_2}), we have
\begin{align}
&\quad\frac{\Pr\big(\mathscr{P}_{\rm D-PPSC}(\betab)=\betab^\sharp\big)}{\Pr\big(\mathscr{P}_{\rm D-PPSC}(\betab^\prime)=\betab^\sharp\big)}\notag\\
&\le \exp\bigg(\frac{\big\|(\Db^\top\Db)^{-1}\Db^\top\Cb(\betab-\betab^\prime)\big\|_1}{v}\bigg)\notag\\
&\le\exp\bigg(\frac{\sqrt{n-1}\big\|(\Db^\top\Db)^{-1}\big\|_2\big\|\Db^\top\big\|_1\big\|\Cb\big\|_1\big\|\betab-\betab^\prime\big\|_1}{v}\bigg).\label{eq:dp_3}
\end{align}
Evidently, $\|\Db^\top\|_1$ is equal to the maximum degree of $\mGM_{\betab^\sharp}$ according to the definition of $\Db$ and $\mGM_{\betab^\sharp}$. Since each column of $\Cb$ is equal to one, $\|\Cb\|_1=1$. Then (\ref{eq:dp_3}) gives
\begin{align}
\frac{\Pr\big(\mathscr{P}_{\rm D-PPSC}(\betab)=\betab^\sharp\big)}{\Pr\big(\mathscr{P}_{\rm D-PPSC}(\betab^\prime)=\betab^\sharp\big)}\le\exp\bigg(\frac{\delta\sqrt{n-1}\Delta(\mGM_{\betab^\sharp})}{v\left|\sigm(\Db^\top\Db)\right|}\bigg),\notag
\end{align}
which completes the proof.

\subsection*{D. Proof of Theorem \ref{thm:dependence}}
We start the proof by showing the sufficiency of Condition (i), (ii) and (iii). Without loss of generality, we assume $(i_p,i_{p+1})\prec(i_p,i_{p-1})$. Let $\asf_{t_p}=(i_p,i_{p+1})$. According to Algorithm (\ref{eq:alg1}), the node states $x_{i_p}(t_p)$ and $x_{i_{p+1}}(t_p)$ are given by
\begin{align}
x_{i_p}(t_p) &= \gamma_{t_p}\notag\\
x_{i_{p+1}}(t_p) &= x_{i_p}(t_p-1) + x_{i_{p+1}}(t_p-1) - \gamma_{t_p}. \notag
\end{align}
As specified by Condition (i), $i_p,i_{p-1},\dots,i_0$ is a directed path, which satisfies Condition (i) of Lemma \ref{lem:RV_transfer}. Condition (ii) gives
\begin{equation}\notag
\asf_{t_p}\prec(i_p,i_{p-1})\prec\dots\prec(i_1,i_0)\preceq\asf_{n-1 },
\end{equation}
satisfying Condition (ii) of Lemma \ref{lem:RV_transfer}. In addition, Condition (iii) is equivalent to Condition (iii) of Lemma \ref{lem:RV_transfer} for path $i_p,i_{p-1},\dots,i_0$. Thus Lemma \ref{lem:RV_transfer} shows that random variable $\gamma_{t_p}$ appears in node state $x_{i_0}(n-1 )$, i.e.,
\begin{equation}\label{eq:i_0_t_p}
d_{i_0 t_p}(n-1 )=1.
\end{equation}
Analogously, three conditions of Lemma \ref{lem:RV_transfer} are met for path $i_{p+1},i_{p+2},\dots,i_l$, which yields
\begin{equation}\label{eq:i_l_t_p}
d_{i_l t_p}(n-1 )=-1.
\end{equation}
Evidently, (\ref{eq:i_0_t_p}) and (\ref{eq:i_l_t_p}) make it sufficient for node states $x_{i_0}(n-1 )$ and $x_{i_l}(n-1 )$ to be dependent.
In the following, we prove the necessity of Condition (i), (ii) and (iii).\\
Necessity of (i). Since $x_{i_0}(n-1 )$ and $x_{i_l}(n-1 )$ are dependent, there exist random variables $d_{i_0 t_p}(n-1 )\gamma_{t_p}$ and $d_{i_l t_p}(n-1 )\gamma_{t_p}$ that appear in node states $x_{i_0}(n-1 )$ and $x_{i_l}(n-1 )$, respectively. Algorithm (\ref{eq:alg1}) gives
\begin{align}
x_{\Tail(\asf_{t_p})}(t_p) &= \gamma_{t_p}\notag\\
x_{\Head(\asf_{t_p})}(t_p) &= x_{\Tail(\asf_{t_p})}(t_p-1) + x_{\Head(\asf_{t_p})}(t_p-1) - \gamma_{t_p}. \notag
\end{align}
We suppose that $\gamma_{t_p}$ and $-\gamma_{t_p}$ transfer to $x_{\Tail(\asf_{t_p})}(n-1 )$ and $x_{\Head(\asf_{t_p})}(n-1 )$, respectively, i.e., $d_{i_0 t_p}=1$ and $d_{i_l t_p}=-1$. Due to the uniqueness of paths in spanning trees, nodes $\Tail(\asf_{t_p})$ and $\Tail(\asf_{t_p})$ are in path $i_0,i_1,\dots,i_l$. Let $i_p=\Tail(\asf_{t_p})$. Condition (i) of Lemma \ref{lem:RV_transfer} shows that $i_p,i_{p-1},\dots,i_0$ and $i_p,i_{p+1},\dots,i_l$ are directed paths. In addition, $i_p\neq i_0$ and $i_p\neq i_l$ because there exists no directed path that connects node $i$ and node $j$. The necessity of Condition (i) can be similarly proved, provided that $-\gamma_{t_p}$ and $\gamma_{t_p}$ transfer to $x_{\Tail(\asf_{t_p})}(n-1 )$ and $x_{\Head(\asf_{t_p})}(n-1 )$, respectively.\\
Necessity of (ii). Without loss of generality, we assume $(i_p,i_{p+1})\prec(i_p,i_{p-1})$. Now we prove the necessity of Condition (ii). Since the random variable $d_{i_0 t_p}(n-1 )\gamma_{t_p}$ transfers from $x_{i_p}(t_p)$ to $x_{i_0}(n-1 )$, Lemma \ref{lem:RV_transfer} provides
\begin{equation}\label{eq:i_p_i_0}
(i_p,i_{p-1})\prec\dots\prec(i_1,i_0).
\end{equation}
Similarly, random variable $d_{i_l t_p}(n-1 )\gamma_{t_p}$ transfers from $x_{i_{p+1}}(t_p)$ to $x_{i_l}(n-1 )$ assures
\begin{equation}\label{eq:i_p+1_i_l}
(i_p,i_{p+1})\prec\dots\prec(i_{l-1},i_l).
\end{equation}
Clearly (\ref{eq:i_p_i_0}) and (\ref{eq:i_p+1_i_l}) shows the necessity of Condition (ii).
Necessity of (iii). We finally prove Condition (iii) is necessary for the dependence result. We have shown above that if the node states $x_{i_0}(n-1 )$ and $x_{i_l}(n-1 )$ are dependent, the random variables $d_{i_0 t_p}(n-1 )\gamma_{t_p}$ and $d_{i_l t_p}(n-1 )\gamma_{t_p}$ transfer to $x_{i_0}(n-1 )$ and $x_{i_l}(n-1 )$, respectively. Thus Condition (iii) of Lemma \ref{lem:RV_transfer} are necessarily met on both path $i_p,i_{p-1},\dots,i_0$ and path $i_{p+1},i_{p+2},\dots,i_l$, which is equivalent to Condition (iii). Now the necessity of all three conditions is proved.

\subsection*{E. Proof of Theorem \ref{thm:dependence_directed_path}}
First we focus on   proving the sufficiency part of the statements. Let $\asf_{t_k}=(i_k,i_{k+1})$ for $k=0,1,\dots,l-1$. According to the algorithm (\ref{eq:alg1}), information transmission occurs on edge $(i_0,i_1)$ at time $t_0$. Then it follows
\begin{align}
x_{i_0}(t_0) &= \gammab_{t_0}\notag\\
x_{i_1}(t_0) &= x_{i_0}(t_0-1) + x_{i_1}(t_0-1) - \gammab_{t_0}. \label{eq:i_1_hold_-gamma_t_0}
\end{align}
It is seen from (\ref{eq:i_1_hold_-gamma_t_0}) that the random variables $\gammab_{t_0}$ and $-\gammab_{t_0}$ are held by the node states $x_{i_0}(t_0)$ and $x_{i_1}(t_0)$, respectively. Evidently, Condition (ii) specifies that the endpoints of all edges not equal to $(i_0,i_1)$ with their tail being node $i_0$ exchange information according to the algorithm (\ref{eq:alg1}) prior to $(i_0,i_1)$. Thus
\begin{equation}\label{eq:i_0_gamma_t_0}
d_{i_0 t_0}(n-1 )=1.
\end{equation}
Since $i_0,i_1,\dots,i_l$ is a directed path, Lemma \ref{lem:RV_transfer} provides that Condition (i) and (ii) guarantees that $-\gammab_{t_0}$ transfers to $x_{i_l}(n-1 )$, i.e.,
\begin{equation}\label{eq:i_l_-gamma_t_0}
d_{i_l t_0}(n-1 )=-1.
\end{equation}
(\ref{eq:i_0_gamma_t_0}) and (\ref{eq:i_l_-gamma_t_0}) clearly show that node states $x_{i_0}(n-1 )$ and $x_{i_l}(n-1 )$ are dependent.

Now we prove the necessity of these two conditions. Suppose $x_{i_0}(n-1 )$ and $x_{i_l}(n-1 )$ are dependent. Then there exist random variables $d_{i_0 t_p}(n-1 )\gammab_{t_p}$ and $d_{i_l t_p}(n-1 )\gammab_{t_p}$ that appear in $x_{i_0}(n-1 )$ and $x_{i_l}(n-1 )$, respectively. It is clear that the nodes $\Tail(\asf_{t_p})$ and $\Head(\asf_{t_p})$ are in the path $i_0,i_1,\dots,i_l$ because of the path uniqueness in spanning trees. However, it is impossible that $\Tail(\asf_{t_p})=i_{\hat{k}}$ and $\Head(\asf_{t_p})=i_{\hat{k}+1}$ for any $\hat{k}\in\{1,\dots,l-1\}$, because $i_{\hat{k}},i_{\hat{k}-1},\dots,i_0$ is necessarily a directed path by Lemma \ref{lem:RV_transfer}. Hence $\Tail(\asf_{t_p})=i_0$ and $\Head(\asf_{t_p})=i_1$. In addition, $d_{i_0 t_p}=1$ and $d_{i_1 t_p}=-1$. Finally the random variable $-\gammab_{t_p}$ becomes a component of $x_{i_l}(n-1 )$. Thus by Lemma \ref{lem:RV_transfer}, it is necessary for Condition (i) and (ii) to hold, which completes the proof.

\subsection*{F. Proof of Theorem \ref{thm:Sigma}}

Lemma \ref{lem:final_node_states} (i) implies that at most $n-1 $ pairs of node states in $x_1(n-1 ),x_2(n-1 ),\dots,x_n(n-1 )$ are dependent. Since it is described in Lemma \ref{lem:final_node_states} (ii) that every pair of dependent final node states possess only one of $\pm\gammab_1,\dots,\pm\gammab_{n-1 }$, there are exactly $n-1 $ pairs of dependent final node states. In addition, it can be calculate the covariance of two dependent states is $-\sigma^2$. Thus $|\mEM_{\betab^\sharp}|=n-1 $. Second, we complete the proof that $\mGM_{\betab^\sharp}$ is a tree by showing $\mGM_{\betab^\sharp}$ is connected. We know for an arbitrary node $i\in\mV$, there exists $t\in\{1,2,\dots,n-1 \}$ such that $i$ is an endpoint of $\asf_t$. Thus the algorithm (\ref{eq:alg1}) guarantees that each $x_i(n-1 )$ holds at least one of $\gammab_1,\dots,\gammab_{n-1 }$ and thereby $\mGM_{\betab^\sharp}$ has no isolated nodes. Assume, for contradiction, that $\mGM_{\betab^\sharp}$ has $r$ connected components with $r>1$. Let $\mGM_{\mathrm{S}_1}$ and $\mGM_{\mathrm{S}_2}$ denote two of the connected components of $\mGM_{\betab^\sharp}$ on $\mS_1,\mS_2\subset\mV$, respectively. Then there exists $s_0\in\{1,2,\dots,n-1 \}$ such that edge $\asf_{s_0}$ has one of its endpoints in $\mS_1$ and the other in $\mS_2$. Without loss of generality, we assume $\Head(\asf_{s_0})\in\mS_1$ and $\Tail(\asf_{s_0})\in\mS_2$. By the definition of connected components, there exists $\mS\in\mV$ such that $\mGM_{\mathrm{S}}$ is a connected component of $\mGM_{\betab^\sharp}$ on $\mathrm{S}$ and there exist node $i,j\in\mS$ such that $x_i(n-1 )$ holds $\gammab_{s_0}$ and $x_j(n-1 )$ holds $-\gammab_{s_0}$, resulting in $\{i,j\}\in\mEM_{\betab^\sharp}$. Similarly to the proof of Lemma \ref{lem:final_node_states} (ii), the dependence of $x_i(n-1 )$ and $x_j(n-1 )$ gives
\begin{enumerate}[(i)]
	\item The path from $\Head(\asf_{s_0})$ to $i$ is a directed path and $\Tail(\asf_{s_0})=j$;
	\item The path from $\Tail(\asf_{s_0})$ to $j$ is a directed path and $\Head(\asf_{s_0})=i$;
	\item The paths from $\Head(\asf_{s_0})$ to $i$ and from $\Tail(\asf_{s_0})$ to $j$ are both directed paths.
\end{enumerate}
In case (i), conditions in Theorem \ref{thm:dependence_directed_path} for $\Head(\asf_{s_0})$ and $i$ are still satisfied as a result of the dependence of $x_i(n-1 ),x_j(n-1 )$ by Theorem \ref{thm:dependence} or Theorem \ref{thm:dependence_directed_path}. As a consequence, $x_{\Head(\asf_{s_0})}(n-1 )$ and $x_i(n-1 )$ are dependent and $\{\Head(\asf_{s_0}),i\}\in\mEM_{\betab^\sharp}$. Clearly, the edges $\{\Head(\asf_{s_0}),i\}$ and $\{i,\Tail(\asf_{s_0})\}$ makes neither $\mGM_{\mathrm{S}_1}$ nor $\mGM_{\mathrm{S}_2}$ not connected components. Hence $\mGM_{\betab^\sharp}$ has one connected component, i.e., $\mGM_{\betab^\sharp}$ is connected. Therefore, $\mGM_{\betab^\sharp}$ is a tree in case (i). The same conclusion can be drawn for case (ii) and (iii). Thus $\mGM_{\betab^\sharp}$ is a tree.\\
Next we show $\Sigmab_{\betab^\sharp}$ is the Laplacian of $\mGM_{\betab^\sharp}$ by proving the following properties of $\Sigmab_{\betab^\sharp}$.
\begin{enumerate}[(i)]
	\item For any $i\neq j$
	\begin{equation}\notag
	[\Sigmab_{\betab^\sharp}]_{ij}=\left\{
	\begin{aligned}
	-\sigma_{\gammab}^2\ \ &\textnormal{ if }\{i,j\}\in\mEM_{\betab^\sharp};\\
	0\ \ &\textnormal{ otherwise}.
	\end{aligned}
	\right.
	\end{equation}
	\item For any $i\in\{1,2,\dots,n\}$
	\begin{equation}\notag
	[\Sigmab_{\betab^\sharp}]_{ii} = -\sum\limits_{j\neq i}[\Sigmab_{\betab^\sharp}]_{ij}.
	\end{equation}
\end{enumerate}
Proof of (i). If $\{i,j\}\in\mEM_{\betab^\sharp}$, by Lemma \ref{lem:final_node_states} (ii), the covariance of $x_i(n-1 )$ and $x_j(n-1 )$ can be calculated $[\Sigmab_{\betab^\sharp}]_{ij}=\cov(\gammab_s,-\gammab_s)=-\var(\gammab_s)=-\sigma^2$ for some $s\in\{1,2,\dots,n-1 \}$. If $\{i,j\}\notin\mEM_{\betab^\sharp}$, $[\Sigmab_{\betab^\sharp}]_{ij}=0$ by its definition. Then (i) has been proved.\\
Proof of (ii). Suppose node $i$ has $r_i$ random variables forming a subset of $\{\pm\gammab_1,\dots,\pm\gammab_{n-1 }\}$. Then $[\Sigmab_{\betab^\sharp}]_{ii}=r_i\sigma^2$. By Lemma \ref{lem:final_node_states} (i) and (ii), there exists $r_i$ node states that are dependent of $x_i(n-1 )$ with covariance $-\sigma^2$. This completes the proof of (ii).\\
It is clear (i) and (ii) show $\Sigmab_{\betab^\sharp}$ is the Laplacian of $\mGM_{\betab^\sharp}$.




%

\end{document}